\documentclass[fleqn]{llncs}
\usepackage{nz}

\title{On the Complexity of the Correctness Problem for Non-Zeroness 
       Test Instruction Sequences}
\author{J.A. Bergstra \and C.A. Middelburg}
\institute{Informatics Institute, Faculty of Science, University of
           Amsterdam \\
           Science Park~904, 1098~XH Amsterdam, the Netherlands \\
           \email{J.A.Bergstra@uva.nl,C.A.Middelburg@uva.nl}}

\begin{document}

\maketitle

\begin{abstract}
This paper concerns the question to what extent it can be efficiently 
determined whether an arbitrary program correctly solves a given 
problem.
This question is investigated with programs of a very simple form, 
namely instruction sequences, and a very simple problem, namely the 
non-zeroness test on natural numbers.
The instruction sequences concerned are of a kind by which, for each 
$n > 0$, each function from $\set{0,1}^n$ to $\set{0,1}$ can 
be computed.
The established results include the time complexities of the problem of 
determining whether an arbitrary instruction sequence correctly 
implements the restriction to $\set{0,1}^n$ of the function from 
$\set{0,1}^*$ to $\set{0,1}$ that models the non-zeroness test function, 
for $n > 0$, under several restrictions on the arbitrary instruction 
sequence.
\begin{keywords} 
non-zeroness test, instruction sequence, correctness problem, 
computational complexity.
\end{keywords}%
\begin{classcode}
D.2.4; F.1.1; F.1.3.
\end{classcode}
\end{abstract}

\section{Introduction}
\label{sect-intro}

For each $n > 0$, each function from $\set{0,1}^n$ to $\set{0,1}$ can be 
computed by a finite instruction sequence that contains only 
instructions to set and get the content of Boolean registers, forward 
jump instructions, and a termination instruction.
It is an intuitively evident fact that the correctness of an arbitrary 
instruction sequence of this kind as an implementation of the 
restriction to $\set{0,1}^n$ of a given function from $\set{0,1}^*$ to 
$\set{0,1}$, for $n > 0$, cannot be efficiently determined.
In this paper, we investigate under what restrictions on the arbitrary 
instruction sequence the correctness can be efficiently determined in 
the case that the given function is the function that models the 
non-zeroness test on natural numbers with respect to their binary 
representations.
To our knowledge, there are no previous investigations of this kind.

One of the main results of this work 
(Theorem~\ref{theorem-complexity-shortest-linear}) states, roughly, that 
the problem of determining the correctness of an arbitrary instruction 
sequence as an implementation of the restriction to $\set{0,1}^n$ of the 
function from $\set{0,1}^*$ to $\set{0,1}$ that models the non-zeroness 
test function, for $n > 0$, is co-NP-complete, even under the 
restriction on the arbitrary instruction sequence that its length 
depends linearly on $n$. 
Another of the main results of this work 
(Theorem~\ref{theorem-complexity-shortest-const-2}) states, roughly, 
that this problem can be decided in time polynomial in $n$ under the 
restriction on the arbitrary instruction sequence that its length is the 
length of the shortest possible correct implementations plus a constant 
amount and that it has a certain form.
We expect that similar results can be established for many other 
functions, but possibly at considerable effort.

The question to what extent it can be efficiently determined whether an 
arbitrary program correctly solves a given problem is of importance to 
programming.
We have chosen to investigate this question but, to our knowledge, there 
does not exist literature about it.
This made us decide to start our investigation with programs of a very 
simple form, namely instruction sequences, and a very simple problem, 
namely the non-zeroness problem.
Moreover, we decided to conduct our investigation as an application of 
program algebra, the algebraic theory of instruction sequences that we 
have developed (see below).

Instruction sequences are programs in their simplest form.
Therefore, it is to be expected that it is somehow easier to understand 
the concept of an instruction sequence than to understand the concept of 
a program. 
The first objective of our work on instruction sequences that started 
with~\cite{BL02a}, and of which an enumeration is available 
at~\cite{SiteIS}, is to understand the concept of a program.
The basis of all this work is an algebraic theory of instruction 
sequences, called program algebra, and an algebraic theory of 
mathematical objects that represent in a direct way the behaviours 
produced by instruction sequences under execution, called basic thread 
algebra.%
\footnote
{Both program algebra and basic thread algebra were first introduced 
 in~\cite{BL02a}, but in that paper the latter was introduced under the 
 name basic polarized process algebra.}
The body of theory developed through this work is such that its use as a
conceptual preparation for programming is practically feasible.

The notion of an instruction sequence appears in the work concerned as 
a mathematical abstraction for which the rationale is based on the 
objective mentioned above. 
In this capacity, instruction sequences constitute a primary field of 
investigation in programming comparable to propositions in logic and 
rational numbers in arithmetic. 
The structure of the mathematical abstraction at issue has been 
determined in advance with the hope of applying it in diverse 
circumstances where in each case the fit may be less than perfect.
Until now, this work has, among other things, yielded an approach to 
computational complexity where program size is used as complexity 
measure, a contribution to the conceptual analysis of the notion of an 
algorithm, and new insights into such diverse issues as the halting 
problem, garbage collection, program parallelization for the purpose of 
explicit multi-threading and virus detection.

Like in the work on computational complexity (see~\cite{BM13a,BM14e}) 
and the work on algorithmic equivalence of programs (see~\cite{BM14a}) 
referred to above, in the work presented in this paper, use is made of 
the fact that, for each $n > 0$, each function from $\set{0,1}^n$ to 
$\set{0,1}$ can be computed by a finite instruction sequence that 
contains only instructions to set and get the content of Boolean 
registers, forward jump instructions, and a termination instruction.
Program algebra is parameterized by a set of uninterpreted basic 
instructions.
In applications of program algebra, this set is instantiated by a set of 
interpreted basic instructions.
In a considerable part of the work belonging to our work on instruction 
sequences that started with~\cite{BL02a}, the interpreted basic 
instructions are instructions to set and get the content of Boolean 
registers.

This paper is organized as follows.
First, we survey program algebra and the particular fragment and 
instantiation of it that is used in this paper (Section~\ref{sect-PGA}).
Next, we present a simple non-zeroness test instruction sequence   
(Section~\ref{sect-ISNZ-natural}).
After that, as a preparation for establishing the main results, we first
present a non-zeroness test instruction sequence whose length is minimal 
(Section~\ref{sect-ISNZ-shortest}) and then introduce the set of all 
non-zeroness test instruction sequences of minimal length 
(Section~\ref{sect-ISNZ-shortest-all}).
Following this, we study the time complexity of several restrictions of 
the problem of deciding whether an arbitrary instruction sequence 
correctly implements the restriction to $\set{0,1}^n$ of the function 
from $\set{0,1}^*$ to $\set{0,1}$ that models the non-zeroness test 
function, for $n > 0$ (Section~\ref{sect-complexity}).
Finally, we make some concluding remarks (Section~\ref{sect-concl}).

As mentioned earlier, to our knowledge, there is no previous work that 
addresses a question similar to the question to what extent it can be 
efficiently determined whether an arbitrary program correctly solves a 
given problem.
For this reason, there is no mention of related work in this paper.

The following should be mentioned in advance.
The set $\Bool$ of Boolean values is a set with two elements whose 
intended interpretations are the truth values \emph{false} and 
\emph{true}.
As is common practice, we represent the elements of $\Bool$ by the bits 
$0$ and $1$ and we identify the elements of $\Bool$ with their 
representation where appropriate.

This paper draws somewhat from the preliminaries of earlier papers that 
built on program algebra.
The most recent one of those papers is~\cite{BM17a}.

\section{Preliminaries: instruction sequences and computation}
\label{sect-PGA}

In this section, we present a brief outline of \PGA\ (ProGram Algebra) 
and the particular fragment and instantiation of it that is used in this 
paper.
A mathematically precise treatment of this particular case can be found 
in~\cite{BM13a}.

The starting-point of \PGA\ is the simple and appealing perception
of a sequential program as a single-pass instruction sequence, i.e., a
finite or infinite sequence of instructions each of which is executed at 
most once and can be dropped after it has been executed or jumped over.

It is assumed that a fixed but arbitrary set $\BInstr$ of
\emph{basic instructions} has been given.
The intuition is that the execution of a basic instruction may modify a 
state and produces a reply at its completion.
The possible replies are $\False$ and $\True$.
The actual reply is generally state-dependent.
Therefore, successive executions of the same basic instruction may
produce different replies.
The set $\BInstr$ is the basis for the set of instructions of which the 
instruction sequences considered in \PGA\ are composed.
The elements of the latter set are called \emph{primitive instructions}.
There are five kinds of primitive instructions, which are listed below:
\begin{itemize}
\item
for each $a \in \BInstr$, a \emph{plain basic instruction} $a$;
\item
for each $a \in \BInstr$, a \emph{positive test instruction} $\ptst{a}$;
\item
for each $a \in \BInstr$, a \emph{negative test instruction} $\ntst{a}$;
\item
for each $l \in \Nat$, a \emph{forward jump instruction} $\fjmp{l}$;
\item
a \emph{termination instruction} $\halt$.
\end{itemize}
We write $\PInstr$ for the set of all primitive instructions.

On execution of an instruction sequence, these primitive instructions
have the following effects:
\begin{itemize}
\item
the effect of a positive test instruction $\ptst{a}$ is that basic
instruction $a$ is executed and execution proceeds with the next
primitive instruction if $\True$ is produced and otherwise the next
primitive instruction is skipped and execution proceeds with the
primitive instruction following the skipped one --- if there is no
primitive instruction to proceed with,
inaction occurs;
\item
the effect of a negative test instruction $\ntst{a}$ is the same as
the effect of $\ptst{a}$, but with the role of the value produced
reversed;
\item
the effect of a plain basic instruction $a$ is the same as the effect
of $\ptst{a}$, but execution always proceeds as if $\True$ is produced;
\item
the effect of a forward jump instruction $\fjmp{l}$ is that execution
proceeds with the $l$th next primitive instruction of the instruction
sequence concerned --- if $l$ equals $0$ or there is no primitive
instruction to proceed with, inaction occurs;
\item
the effect of the termination instruction $\halt$ is that execution
terminates.
\end{itemize}
Inaction occurs if no more basic instructions are executed, but 
execution does not terminate.

To build terms, \PGA\ has a constant for each primitive instruction and 
two operators. 
These operators are: the binary concatenation operator ${} \conc {}$ and 
the unary repetition operator ${}\rep$.
We use the notation $\Conc{i = k}{n} P_i$, 
where $k \leq n$ and $P_k,\ldots,P_n$ are \PGA\ terms, 
for the \PGA\ term $P_k \conc \ldots \conc P_n$.
We use the convention that $P \conc \Conc{i = k}{n} P_i$ and 
$\Conc{i = k}{n} P_i \conc P$ stand for $P$ if $k > n$.

The instruction sequences that concern us in the remainder of this paper 
are the finite ones, i.e., the ones that can be denoted by \PGA\ terms 
without variables in which the repetition operator does not occur. 
Moreover, the basic instructions that concern us are instructions to set 
and get the content of Boolean registers.
More precisely, we take the set
\begin{ldispl}
\set{\inbr{i}.\getbr \where i \in \Natpos} \union
\set{\outbr.\setbr{b} \where b \in \Bool}
\\ \;\; {} \union
\set{\auxbr{i}.\getbr \where i \in \Natpos} \union
\set{\auxbr{i}.\setbr{b} \where i \in \Natpos \Land b \in \Bool}
\end{ldispl}%
as the set $\BInstr$ of basic instructions.%
\footnote
{We write $\Natpos$ for the set $\set{n \in \Nat \where n \geq 1}$ of
positive natural numbers.}

Each basic instruction consists of two parts separated by a dot.
The part on the left-hand side of the dot plays the role of the name of 
a Boolean register and the part on the right-hand side of the dot plays 
the role of a command to be carried out on the named Boolean register.
The names are employed as follows:
\begin{itemize}
\item
for each $i \in \Natpos$,
$\inbr{i}$ serves as the name of the Boolean register that is used as 
$i$th input register in instruction sequences;
\item
$\outbr$ serves as the name of the Boolean register that is used as
output register in instruction sequences;
\item
for each $i \in \Natpos$,
$\auxbr{i}$ serves as the name of the Boolean register that is used as 
$i$th auxiliary register in instruction sequences.
\end{itemize}
On execution of a basic instruction, the commands have the following 
effects:
\begin{itemize}
\item
the effect of $\getbr$ is that nothing changes and the reply is the 
content of the named Boolean register;
\item
the effect of $\setbr{\False}$ is that the content of the named Boolean 
register becomes $\False$ and the reply is $\False$;
\item
the effect of $\setbr{\True}$ is that the content of the named Boolean 
register becomes $\True$ and the reply is $\True$.
\end{itemize}

We write $\ISbr$ for the set of all instruction sequences that can 
be denoted by a \PGA\ term without variables in which the repetition 
operator does not occur in the case that $\BInstr$ is taken as specified 
above.
$\ISbr$ is the set of all instruction sequences that matter in the 
remainder of this paper.

We write $\len(X)$, where $X \in \ISbr$, for the length of $X$.

Let $n \in \Nat$, let $\funct{f}{\Bool^n}{\Bool}$, and 
let $X \in \ISbr$.
Then $X$ \emph{computes} $f$ if there exists a $k \in \Nat$ such that, 
for all $b_1,\ldots,b_n \in \Bool$, on execution of $X$ in an 
environment with input registers $\inbr{1},\ldots,\inbr{n}$, output 
register $\outbr$, and auxiliary registers $\auxbr{1},\ldots,\auxbr{k}$,
if 
\begin{itemize} 
\item 
for each $i \in \set{1,\ldots,n}$, 
the content of register $\inbr{i}$ is $b_i$ when execution starts;
\item
the content of register $\outbr$ is $\False$ when execution starts; 
\item 
for each $i \in \set{1,\ldots,k}$, 
the content of register $\auxbr{i}$ is $\False$ when execution starts;
\end{itemize}
then execution terminates and the content of register $\outbr$ is 
$f(b_1,\ldots,b_n)$ when execution terminates.

We conclude these preliminaries with some terminology and notations that 
are used in the rest of this paper.

We refer to the content of a register when execution starts as the 
\emph{initial} content of the register and we refer to the content of 
a register when execution terminates as the \emph{final} content of 
the register.

The primitive instructions of the forms $\ptst{\inbr{i}.\getbr}$ and 
$\ntst{\inbr{i}.\getbr}$ are called \emph{read instructions}.
For a read instruction $u$, the input register whose name appears in $u$ 
is said to be the input register that is \emph{read by} $u$.
For an $X \in \ISbr$ and $m \in \Natpos$, an input register that is read 
by $m$ occurrences of a read instruction in $X$ is said to be an input 
register that is \emph{read $m$ times in} $X$.
For an $X \in \ISbr$, an occurrence of a read instruction in $X$ that is 
neither immediately preceded nor immediately followed by a read 
instruction is said to be an \emph{isolated read instruction} of $X$.
For an $X \in \ISbr$, an occurrence of two read instructions in a row in 
$X$ that is neither immediately preceded nor immediately followed by a 
read instruction is said to be a \emph{read instruction pair} of $X$.
We write $\iregs(X)$, where $X \in \ISbr$, for the set of all 
$i \in \Natpos$ such that $\inbr{i}$ is read by some occurrence of a 
read instruction in $X$.

Take an instruction sequence $X \in \ISbr$ and a function 
$\funct{f}{\Bool^n}{\Bool}$ ($n \in \Nat$) such that $X$ computes $f$.
Modify $X$ by replacing all occurrences of $\ptst{\inbr{i}.\getbr}$ by 
$\fjmp{1}$, all occurrences of $\ntst{\inbr{i}.\getbr}$ by $\fjmp{2}$, 
and, for each $j > i$, all occurrences of the register name $\inbr{j}$ 
by $\inbr{j{-}1}$.
Then the resulting instruction sequence computes the function 
$\funct{f'}{\Bool^{n-1}}{\Bool}$ defined by 
$f'(x_1,\ldots,x_{n-1}) =
 f(x_1,\ldots,x_{i-1},1,x_{i+1},\ldots,x_{n-1})$.
If the occurrences of $\ptst{\inbr{i}.\getbr}$ are replaced by 
$\fjmp{2}$ instead of $\fjmp{1}$ and the occurrences of 
$\ntst{\inbr{i}.\getbr}$ is replaced by $\fjmp{1}$ instead of 
$\fjmp{2}$, then the resulting instruction sequence computes the 
function $\funct{f''}{\Bool^{n-1}}{\Bool}$ defined by 
$f''(x_1,\ldots,x_{n-1}) =
 f(x_1,\ldots,x_{i-1},0,x_{i+1},\ldots,x_{n-1})$.
Such register elimination and its generalization from one register to 
multiple registers are used a number of times in this paper.
A notation for register elimination is introduced in the next paragraph.

For an $X \in \ISbr$ and a function $\alpha$ from a finite subset of
$\Natpos$ to $\Bool$ such that
$\iregs(X) = \set{i \in \Natpos \where i \leq n}$ for some 
$n \in \Natpos$ and $\dom(\alpha)$ is a proper subset of $\iregs(X)$,  
we write $X_\alpha$ for the instruction sequence obtained from $X$ 
by replacing, for each $i \in \dom(\alpha)$,  
all occurrences of $\ptst{\inbr{i}.\getbr}$ by $\fjmp{1}$ if 
$\alpha(i) = 1$ and by $\fjmp{2}$ if $\alpha(i) = 0$,
all occurrences of $\ntst{\inbr{i}.\getbr}$ by $\fjmp{2}$ if 
$\alpha(i) = 1$ and by $\fjmp{1}$ if $\alpha(i) = 0$, and, 
for each $j \in \iregs(X) \diff \dom(\alpha)$, 
all occurrences of the register name $\inbr{j}$ by $\inbr{\beta(j)}$,
where $\beta$ is the unique bijection from 
$\iregs(X) \diff \dom(\alpha)$ to 
$\set{i \in \Natpos \where 
 i \leq n - \card(\dom(\alpha))}$ such that, 
for all $i,j \in \iregs(X) \diff \dom(\alpha)$ with $i \leq j$, 
$\beta(i) \leq \beta(j)$.
For an $X \in \ISbr$ and an $i \in \Natpos$ such that
$\iregs(X) = \linebreak[2] \set{i \in \Natpos \where i \leq n}$ for some 
$n \in \Natpos$ and $i \in \iregs(X)$,  
we write $X[\inbr{i} = b]$ for $X_\alpha$, where $\alpha$ is the 
function from $\set{i}$ to $\Bool$ defined by $\alpha(i) = b$.

Register elimination is reminiscent of gate elimination as used in
much work on circuit lower bounds (see e.g.\ Chapter~16 of~\cite{Weg05a}).

\section{Simple non-zeroness test instruction sequences}
\label{sect-ISNZ-natural}

The remainder of the paper goes into programming by means of instruction 
sequences of the kind introduced in Section~\ref{sect-PGA}. 
We consider the programming of a function from $\Bool^*$ to $\Bool$ that 
models a particular function from $\Nat$ to $\Bool$ with respect to 
the binary representations of the natural numbers by elements from 
$\Bool^*$.
The particular function is the \emph{non-zeroness test} function 
$\tstnza$ defined by the equations $\tstnza(0) = 0$ and 
$\tstnza(k + 1) = 1$.
In this section, we present a simple instruction sequence computing the 
restriction to $\Bool^n$ of the function from $\Bool^*$ to $\Bool$ that 
models this function, for $n > 0$.

$\tstnzc{n}$, the restriction to $\Bool^n$ of the function from 
$\Bool^*$ to $\Bool$ that models $\tstnza$, is defined by
\begin{ldispl}
\tstnzc{n}(b_1,\ldots,b_n) = 1 \;\;\mathrm{iff}\;\;
b_1 = 1 \;\mathrm{or}\; \ldots \;\mathrm{or}\; b_n = 1\;.
\end{ldispl}%

We define an instruction sequence $\TSTNZ{n}$ which is intended to 
compute $\tstnzc{n}$ as follows:
\begin{ldispl}
\TSTNZ{n} \deq 
\Conc{i = 1}{n} (\ptst{\inbr{i}.\getbr} \conc \outbr.\setbr{\True})
 \conc \halt\;.
\end{ldispl}%

The following proposition states that the instruction sequence 
$\TSTNZ{n}$ correctly implements $\tstnzc{n}$.
\begin{proposition}
\label{prop-NZTIS-correct}
For each $n \in \Natpos$, $\TSTNZ{n}$ computes $\tstnzc{n}$.
\end{proposition}
\begin{proof}
We prove this proposition by induction on $n$.
The basis step consists of proving that $\TSTNZ{1}$ computes 
$\tstnzc{1}$. 
This follows easily by a case distinction on the content of $\inbr{1}$. 
The inductive step is proved in the following way.
It follows directly from the induction hypothesis that on execution of 
$\TSTNZ{n+1}$, after 
$\Conc{i = 1}{n} (\ptst{\inbr{i}.\getbr} \conc \outbr.\setbr{\True})$ 
has been executed, 
(a)~the content of $\outbr$ equals $1$ iff the content of at least one 
of the input registers $\inbr{1},...,\inbr{n}$ equals $1$ and 
(b)~execution proceeds with the next instruction.
From this, it follows easily by a case distinction on the content of 
$\inbr{n{+}1}$ that $\TSTNZ{n+1}$ computes $\tstnzc{n+1}$. 
\qed
\end{proof}
 
The length of the instruction sequence $\TSTNZ{n}$ defined above is as 
follows:
\begin{ldispl}
\len(\TSTNZ{n}) = 2 \mul n + 1\;.
\end{ldispl}%

$\TSTNZ{n}$ is a simple instruction sequence to compute $\tstnzc{n}$.
It computes $\tstnzc{n}$ by checking all input registers.
This is rather inefficient because, once an input register is 
encountered whose content is $1$, checking of the remaining input 
registers can be skipped.
$\TSTNZ{n}$ does, moreover, not belong to the shortest instruction 
sequences computing $\tstnzc{n}$.
The shortest instruction sequences computing $\tstnzc{n}$ are the 
subject of Sections~\ref{sect-ISNZ-shortest} 
and~\ref{sect-ISNZ-shortest-all}.

\section{Shortest non-zeroness test instruction sequences}
\label{sect-ISNZ-shortest}

For $i \in \Natpos$, we have that execution of the instruction sequences 
denoted by 
$\ptst{\inbr{i}.\getbr} \conc \outbr.\setbr{\True} \conc 
 \ptst{\inbr{i{+}1}.\getbr} \conc \outbr.\setbr{\True} \conc \halt$ 
and 
$\ntst{\inbr{i}.\getbr} \conc
 \ptst{\inbr{i{+}1}.\getbr} \conc \outbr.\setbr{\True} \conc \halt$ 
yield the same final content of $\outbr$ for all initial contents of 
$\inbr{i}$ and $\inbr{i{+}1}$.
In this section, we present an instruction sequence $\TSTNZp{n}$ which 
can be considered an adaptation of $\TSTNZ{n}$ based on this fact.
There are no instruction sequences shorter than $\TSTNZp{n}$ that 
compute $\tstnzc{n}$.
Section~\ref{sect-ISNZ-shortest-all} is concerned with the set of all
instruction sequences of the same length as $\TSTNZp{n}$ that compute 
$\tstnzc{n}$.

We define an instruction sequence $\TSTNZp{n}$ which is intended to 
compute $\tstnzc{n}$ as follows:
\begin{ldispl}
\TSTNZp{n} \deq {}
\\ \quad
\left\{
\begin{array}{l@{}}
\Conc{i = 1}{n/2} 
 (\ntst{\inbr{2{\mul}i{-}1}.\getbr} \conc
  \ptst{\inbr{2{\mul}i}.\getbr} \conc \outbr.\setbr{\True}) \conc 
\halt \\
\mathrm{if} \; n \mathrm{\;is\;even},
\eqnsep
\ptst{\inbr{1}.\getbr} \conc \outbr.\setbr{\True} \conc
\Conc{i = 1}{(n-1)/2} 
 (\ntst{\inbr{2{\mul}i}.\getbr} \conc
  \ptst{\inbr{2{\mul}i{+}1}.\getbr} \conc \outbr.\setbr{\True}) \conc 
\halt  \\
\mathrm{if} \; n \mathrm{\;is\;odd}.
\end{array}
\right.
\end{ldispl}%

The following proposition states that the instruction sequence 
$\TSTNZp{n}$ correctly implements $\tstnzc{n}$.
\begin{proposition}
\label{prop-NZTISp-correct}
For each $n \in \Natpos$, $\TSTNZp{n}$ computes $\tstnzc{n}$.
\end{proposition}
\begin{proof}
We split the proof of this proposition into a proof for even $n$ and a 
proof for odd $n$.
The proof for even $n$ goes by induction on $n$.
The basis step consists of proving that $\TSTNZp{2}$ computes 
$\tstnzc{2}$. 
This follows easily by a case distinction on the contents of $\inbr{1}$ 
and $\inbr{2}$. 
The inductive step is proved in the following way.
It follows directly from the induction hypothesis that on execution of 
$\TSTNZp{n+2}$, after 
$\Conc{i = 1}{n/2} 
 (\ntst{\inbr{2{\mul}i{-}1}.\getbr} \conc
  \ptst{\inbr{2{\mul}i}.\getbr} \conc \outbr.\setbr{\True})$ 
has been executed, 
(a)~the content of $\outbr$ equals $1$ iff the content of at least one 
of the input registers $\inbr{1},...,\inbr{n}$ equals $1$ and 
(b)~execution proceeds with the next instruction.
From this, it follows easily by a case distinction on the contents of 
$\inbr{n{+}1}$ and $\inbr{n{+}2}$ that $\TSTNZp{n+2}$ computes 
$\tstnzc{n+2}$. 
The proof for odd $n$ is similar.
\qed
\end{proof}

The length of the instruction sequence $\TSTNZp{n}$ defined above is as 
follows:
\begin{ldispl}
\len(\TSTNZp{n}) = 
\left\{
\begin{array}{l@{\quad}l}
3 \mul \displaystyle\frac{n}{2} + 1 & 
\mathrm{if} \; n \mathrm{\;is\;even},
\eqnsep
3 \mul \displaystyle\frac{n+1}{2}   & 
\mathrm{if} \; n \mathrm{\;is\;odd}. 
\end{array}
\right.
\end{ldispl}%

\begin{proposition}
\label{prop-NZTISp-length}
For each $n \in \Natpos$, we have
$\len(\TSTNZp{1}) = 3$,
$\len(\TSTNZp{n+1}) = \len(\TSTNZp{n}) + 2$ if $n$ is even, and
$\len(\TSTNZp{n+1}) = \len(\TSTNZp{n}) + 1$ if $n$ is odd.
\end{proposition}
\begin{proof}
This follows immediately from the fact that
$\len(\TSTNZp{n}) = 3 \mul n / 2  + 1$ if $n$ is even
and
$\len(\TSTNZp{n}) = 3 \mul (n+1) / 2$ if $n$ is odd. 
\qed
\end{proof}
Proposition~\ref{prop-NZTISp-length} and the following corollary of this
proposition are used in several proofs to come.
\begin{corollary}
\label{corollary-NZTISp-length}
We have
$\len(\TSTNZp{n}) + 1 \leq \len(\TSTNZp{n+1}) \leq \len(\TSTNZp{n}) + 2$
and $\len(\TSTNZp{n+2}) = \len(\TSTNZp{n}) + 3$.
\end{corollary}

We also have $\len(\TSTNZp{1}) = \len(\TSTNZ{1})$ and
$\len(\TSTNZp{n}) < \len(\TSTNZ{n})$ for each $n > 1$.
In fact, $\TSTNZp{n}$ belongs to the shortest instruction sequences 
computing $\tstnzc{n}$.
This is stated by the following theorem.
\begin{theorem}
\label{theorem-NZTISp-shortest}
For all $n \in \Natpos$, for all $X \in \ISbr$, $X$ computes 
$\tstnzc{n}$ only if $\len(X) \geq \len(\TSTNZp{n})$.
\end{theorem}
\begin{proof}
We prove the following stronger result:
\begin{quote}
\emph
{for all $n \in \Natpos$, for all $X \in \ISbr$, 
 $X$ or $\outbr.\setbr{1} \conc X$ computes $\tstnzc{n}$ \linebreak[2]
 only if $\len(X) \geq \len(\TSTNZp{n})$.}
\end{quote}

\sloppy
We use the following in the proof. 
Let $\chi_i(X)$, where $i \in \Natpos$, be obtained from $X$ by 
replacing all occurrences of 
$\ptst{\auxbr{i}.\getbr}$ and $\ntst{\auxbr{i}.\getbr}$ by 
$\ntst{\auxbr{i}.\getbr}$ and $\ptst{\auxbr{i}.\getbr}$, respectively,
and replacing, for each $b \in \Bool$, all occurrences of 
$\auxbr{i}.\setbr{b}$, $\ptst{\auxbr{i}.\setbr{b}}$, and 
$\ntst{\auxbr{i}.\setbr{b}}$ by 
$\auxbr{i}.\setbr{\Compl{b}}$, $\ntst{\auxbr{i}.\setbr{\Compl{b}}}$, and 
$\ptst{\auxbr{i}.\setbr{\Compl{b}}}$, respectively.%
\footnote{Here, we write $\Compl{b}$ for the complement of $b$.}
It follows directly from the proof of Theorem~8.1 from~\cite{BM14a} that 
$\chi_i(X)$ computes $\tstnzc{n}$ if $X$ computes $\tstnzc{n}$ and $X$ 
is of the form $u \conc Y$ or $\outbr.\setbr{1} \conc u \conc Y$, where 
$u$ is $\auxbr{i}.\setbr{1}$, $\ptst{\auxbr{i}.\setbr{1}}$ or 
$\ntst{\auxbr{i}.\setbr{1}}$.
Moreover, $\len(\chi_i(X)) = \len(X)$.

We prove the theorem by induction on $n$.

The basis step consists of proving that for all $X \in \ISbr$, $X$ or 
$\outbr.\setbr{1} \conc X$ computes $\tstnzc{1}$ only if 
$\len(X) \geq 3$.
The following observations can be made about all $X \in \ISbr$ such that 
$X$ or $\outbr.\setbr{1} \conc X$ computes $\tstnzc{1}$: 
(a)~there must be at least one occurrence of $\ptst{\inbr{1}.\getbr}$ or 
$\ntst{\inbr{1}.\getbr}$ in $X$
--- because otherwise the final content of $\outbr$ will not be 
dependent on the content of $\inbr{1}$;
(b)~there must be at least one occurrence of $\outbr.\setbr{1}$, 
$\ptst{\outbr.\setbr{1}}$ or $\ntst{\outbr.\setbr{1}}$ in $X$ if 
$X$ computes $\tstnzc{1}$ and
    there must be at least one occurrence of $\outbr.\setbr{0}$, 
$\ptst{\outbr.\setbr{0}}$ or $\ntst{\outbr.\setbr{0}}$ in $X$ otherwise 
--- because otherwise the final content of $\outbr$ will always be the 
same;
(d)~there must be at least one occurrence of $\halt$ in $X$ --- because 
otherwise nothing will ever be computed.
It follows trivially from these observations that, for all 
$X \in \ISbr$, $X$ or $\outbr.\setbr{1} \conc X$ computes $\tstnzc{1}$ 
only if $\len(X) \geq 3$.

The inductive step is proved by contradiction.
Suppose that $X \in \ISbr$, $X$ or $\outbr.\setbr{1} \conc X$ computes 
$\tstnzc{n+1}$, and $\len(X) < \len(\TSTNZp{n+1})$.
Assume that there does not exist an $X'\in \ISbr$ such that $X'$ or 
$\outbr.\setbr{1} \conc X'$ computes $\tstnzc{n+1}$ and 
$\len(X') < \len(X)$.
Obviously, this assumption can be made without loss of generality.
From this assumption, it follows that $X = u_1 \conc \ldots \conc u_k$ 
where $k < \len(\TSTNZp{n+1})$, $u_1,\ldots,u_k \in \PInstr$, and 
$u_1$ is $\ptst{\inbr{i}.\getbr}$ or $\ntst{\inbr{i}.\getbr}$ for some
$i \in \Natpos$ such that $i \leq n + 1$.
This can be seen as follows:
\begin{plist}
\item
if $u_1$ is $\halt$ or $\fjmp{l}$ with $l = 0$ or $l \geq k$, then $X$ 
and $\outbr.\setbr{1} \conc X$ cannot compute $\tstnzc{n+1}$;
\item
if $u_1$ is $\fjmp{l}$ with $0 < l < k$, then there is an $X' \in \ISbr$ 
such that $X'$ or $\outbr.\setbr{1} \conc X'$ computes $\tstnzc{n+1}$ 
and $\len(X') < \len(X)$ --- which contradicts the assumption;
\item
if $u_1$ is $\outbr.\setbr{0}$, $\ptst{\outbr.\setbr{0}}$ or 
$\ntst{\outbr.\setbr{0}}$, then $u_1$ can be replaced by $\fjmp{1}$ or 
$\fjmp{2}$ in $X$ and so, there is an $X' \in \ISbr$ such that $X'$ 
computes $\tstnzc{n+1}$ and $\len(X') < \len(X)$ --- which contradicts 
the assumption;
\item
if $u_1$ is $\outbr.\setbr{1}$, $\ptst{\outbr.\setbr{1}}$ or 
$\ntst{\outbr.\setbr{1}}$, then $u_1$ can be replaced by $\fjmp{1}$ or 
$\fjmp{2}$ in $\outbr.\setbr{1} \conc X$ and so, there is an 
$X' \in \ISbr$ such that $\outbr.\setbr{1} \conc X'$ computes 
$\tstnzc{n+1}$ and $\len(X') < \len(X)$ --- which contradicts the 
assumption;
\item
if $u_1$ is $\auxbr{j}.\setbr{0}$, $\ptst{\auxbr{j}.\setbr{0}}$, 
$\ntst{\auxbr{j}.\setbr{0}}$, $\auxbr{j}.\getbr$, 
$\ptst{\auxbr{j}.\getbr}$ or $\ntst{\auxbr{j}.\getbr}$ for some 
$j \in \Natpos$, then $u_1$ can be replaced by $\fjmp{1}$ or $\fjmp{2}$ 
in $X$ and $\outbr.\setbr{1} \conc X$ and so there is an $X' \in \ISbr$ 
such that $X'$ or $\outbr.\setbr{1} \conc X'$ computes $\tstnzc{n+1}$ 
and $\len(X') < \len(X)$ --- which contradicts the assumption;
\item
if $u_1$ is $\auxbr{j}.\setbr{1}$, $\ptst{\auxbr{j}.\setbr{1}}$ or 
$\ntst{\auxbr{j}.\setbr{1}}$ for some $j \in \Natpos$, then 
$\chi_j(u_1)$ can be replaced by $\fjmp{1}$ or $\fjmp{2}$ in $\chi_j(X)$ 
and $\outbr.\setbr{1} \conc \chi_j(X)$ and so, because $\chi_j(X)$ or 
$\outbr.\setbr{1} \conc \chi_j(X)$ also computes $\tstnzc{n+1}$ and 
$\len(\chi_j(X)) = \len(X)$, there is an $X' \in \ISbr$ such that $X'$ 
or $\outbr.\setbr{1} \conc X'$ computes $\tstnzc{n+1}$ and 
$\len(X') < \len(X)$ --- which contradicts the assumption;
\item
if $u_1$ is $\inbr{j}.\getbr$ for some $j \in \Natpos$, then $u_1$ can 
be replaced by $\fjmp{1}$ in $X$ and $\outbr.\setbr{1} \conc X$ and so 
there is an $X' \in \ISbr$ such that $X'$ or $\outbr.\setbr{1} \conc X'$ 
computes $\tstnzc{n+1}$ and $\len(X') < \len(X)$ --- which contradicts 
the assumption;
\item
if $u_1$ is $\ptst{\inbr{j}.\getbr}$ or $\ntst{\inbr{j}.\getbr}$ for 
some $j \in \Natpos$ such that $j > n + 1$, then, because the final 
content of $\outbr$ is independent of the initial content of $\inbr{j}$, 
$u_1$ can be replaced by $\fjmp{1}$ and $\fjmp{2}$ in $X$ and so there 
is an $X' \in \ISbr$ such that $X'$ or $\outbr.\setbr{1} \conc X'$ 
computes $\tstnzc{n+1}$ and $\len(X') < \len(X)$ --- which contradicts 
the assumption.
\end{plist}
So, we distinguish between 
the case that $u_1$ is $\ptst{\inbr{i}.\getbr}$ and 
the case that $u_1$ is $\ntst{\inbr{i}.\getbr}$.

In the case that $u_1$ is $\ptst{\inbr{i}.\getbr}$, we consider the case 
that $\inbr{i}$ contains $0$.
In this case, after execution of $u_1$, execution proceeds with $u_3$.
Let $Y = (u_3 \conc \ldots \conc u_k)\linebreak[2][\inbr{i} = 0]$.
Then $Y$ or $\outbr.\setbr{1} \conc Y$ computes $\tstnzc{n}$.
Moreover, by Corollary~\ref{corollary-NZTISp-length}, we have that 
$\len(Y) = \len(X) - 2 < \len(\TSTNZp{n})$.
Hence, there exists a $Y' \in \ISbr$ such that $Y'$ or 
$\outbr.\setbr{1} \conc Y'$ computes $\tstnzc{n}$ and 
$\len(Y') < \len(\TSTNZp{n})$.
This contradicts the induction hypothesis.

In the case that $u_1$ is $\ntst{\inbr{i}.\getbr}$, we consider the case 
that $\inbr{i}$ contains $0$.
In this case, after execution of $u_1$, execution proceeds with $u_2$.
Let $Y = (u_2 \conc \ldots \conc u_k)\linebreak[2][\inbr{i} = 0]$.
Then $Y$ or $\outbr.\setbr{1} \conc Y$ computes $\tstnzc{n}$. 
From here, because $\len(Y) = \len(X) - 1$, we cannot derive a 
contradiction immediately as in the case that $u_1$ is 
$\ptst{\inbr{i}.\getbr}$.
A case distinction on $u_2$ is needed.
With the exception of the cases that $u_2$ is $\ptst{\inbr{j}.\getbr}$ 
or $\ntst{\inbr{j}.\getbr}$, for some $j \in \Natpos$ such that 
$i \neq j$ and $j \leq n + 1$, we still consider the case that 
$\inbr{i}$ contains $0$.
In the cases that are not excepted above, a contradiction is derived as 
follows:
\begin{plist}
\item
if $u_2$ is $\halt$ or $\fjmp{l}$ with $l = 0$ or $l \geq k - 1$, then 
$Y$ and $\outbr.\setbr{1} \conc Y$ cannot compute $\tstnzc{n}$;
\item
if $u_2$ is $\fjmp{l}$ with $0 < l < k - 1$, then there is a 
$Y' \in \ISbr$ such that $Y'$ or $\outbr.\setbr{1} \conc Y'$ computes 
$\tstnzc{n}$ and, by Corollary~\ref{corollary-NZTISp-length}, 
$\len(Y') < \len(\TSTNZp{n})$ --- which contradicts the induction 
hypothesis;
\item
if $u_2$ is $\outbr.\setbr{0}$, $\ptst{\outbr.\setbr{0}}$ or 
$\ntst{\outbr.\setbr{0}}$, then $u_2$ can be replaced by $\fjmp{1}$ or 
$\fjmp{2}$ in $Y$ and so, there is a $Y' \in \ISbr$ such that $Y'$ 
computes $\tstnzc{n}$ and, by Corollary~\ref{corollary-NZTISp-length}, 
$\len(Y') < \len(\TSTNZp{n})$ --- which contradicts the induction 
hypothesis;   
\item
if $u_2$ is $\outbr.\setbr{1}$, $\ptst{\outbr.\setbr{1}}$ or 
$\ntst{\outbr.\setbr{1}}$, then $u_2$ can be replaced by $\fjmp{1}$ or 
$\fjmp{2}$ in $\outbr.\setbr{1} \conc Y$ and so, there is a 
$Y' \in \ISbr$ such that $\outbr.\setbr{1} \conc Y'$ computes 
$\tstnzc{n}$ and, by Corollary~\ref{corollary-NZTISp-length}, 
$\len(Y') < \len(\TSTNZp{n})$ --- which contradicts the induction 
hypothesis; 
\item
if $u_2$ is $\auxbr{j'}.\setbr{0}$, $\ptst{\auxbr{j'}.\setbr{0}}$, 
$\ntst{\auxbr{j'}.\setbr{0}}$, $\auxbr{j'}.\getbr$, 
$\ptst{\auxbr{j'}.\getbr}$ or $\ntst{\auxbr{j'}.\getbr}$ for some 
$j' \in \Natpos$, then $u_2$ can be replaced by $\fjmp{1}$ or $\fjmp{2}$ 
in $Y$ and so, there is a $Y' \in \ISbr$ such that $Y'$ or 
$\outbr.\setbr{1} \conc Y'$ computes $\tstnzc{n}$ and, by 
Corollary~\ref{corollary-NZTISp-length}, $\len(Y') < \len(\TSTNZp{n})$ 
--- which contradicts the induction hypothesis; 
\item
if $u_2$ is $\auxbr{j'}.\setbr{1}$, $\ptst{\auxbr{j'}.\setbr{1}}$ or 
$\ntst{\auxbr{j'}.\setbr{1}}$ for some $j' \in \Natpos$, then 
$\chi_{j'}(u_2)$ can be replaced by $\fjmp{1}$ or $\fjmp{2}$ in 
$\chi_{j'}(Y)$ and so, because $\chi_{j'}(Y)$ or 
$\outbr.\setbr{1} \conc \chi_{j'}(Y)$ also computes $\tstnzc{n}$ and 
$\len(\chi_{j'}(Y)) = \len(Y)$, there is a $Y' \in \ISbr$ such that 
$Y'$ or $\outbr.\setbr{1} \conc Y'$ computes $\tstnzc{n}$ and, by 
Corollary~\ref{corollary-NZTISp-length}, $\len(Y') < \len(\TSTNZp{n})$ 
--- which contradicts the induction hypothesis;
\item
if $u_2$ is $\inbr{j'}.\getbr$ for some $j' \in \Natpos$, then $u_2$ can 
be replaced by $\fjmp{1}$ in $Y$ and $\outbr.\setbr{1} \conc Y$ and so 
there is an $Y' \in \ISbr$ such that $Y'$ or $\outbr.\setbr{1} \conc Y'$ 
computes $\tstnzc{n}$ and $\len(Y') < \len(\TSTNZp{n})$ --- which 
contradicts the induction hypothesis;
\item
if $u_2$ is $\ptst{\inbr{j'}.\getbr}$ or $\ntst{\inbr{j'}.\getbr}$ for 
some $j' \in \Natpos$ such that $j' > n + 1$, then, because the final 
content of $\outbr$ is independent of the initial content of 
$\inbr{j'}$, $u_2$ can be replaced by $\fjmp{1}$ and $\fjmp{2}$ in $Y$ 
and so there is an $Y' \in \ISbr$ such that $Y'$ or 
$\outbr.\setbr{1} \conc Y'$ computes $\tstnzc{n}$ and, by 
Corollary~\ref{corollary-NZTISp-length}, $\len(Y') < \len(\TSTNZp{n})$ 
--- which contradicts the induction hypothesis;
\item
if $u_2$ is $\ptst{\inbr{i}.\getbr}$ or $\ntst{\inbr{i}.\getbr}$, then 
$u_2$ has been replaced by $\fjmp{1}$ or $\fjmp{2}$ in $Y$ and so 
there is a $Y' \in \ISbr$ such that $Y'$ or $\outbr.\setbr{1} \conc Y'$ 
computes $\tstnzc{n}$ and, by Corollary~\ref{corollary-NZTISp-length}, 
$\len(Y') < \len(\TSTNZp{n})$ --- which contradicts the induction 
hypothesis.
\end{plist}
In the case that $u_2$ is $\ptst{\inbr{j}.\getbr}$, we consider the case 
that both $\inbr{i}$ and $\inbr{j}$ contain $0$.
Let $Z = (u_4 \conc \ldots \conc u_k)[\inbr{i} = 0][\inbr{j} = 0]$.
Then, $Z$ or $\outbr.\setbr{1} \conc Z$ computes $\tstnzc{n-1}$ and, 
by Corollary~\ref{corollary-NZTISp-length}, 
$\len(Z) < \len(\TSTNZp{n-1})$ --- which contradicts the induction 
hypothesis.
In the case that $u_2$ is $\ntst{\inbr{j}.\getbr}$, we consider the case 
that only $\inbr{j}$ contains $0$.
Let $Z' = (u_3 \conc \ldots \conc u_k)[\inbr{j} = 0]$.
Then, $Z'$ or $\outbr.\setbr{1} \conc Z'$ computes $\tstnzc{n}$ and,
by Corollary~\ref{corollary-NZTISp-length}, 
$\len(Z') < \len(\TSTNZp{n})$ --- which contradicts the induction 
hypothesis.
\qed
\end{proof}
Theorem~\ref{theorem-NZTISp-shortest} is a result similar to certain
results on circuit lower bounds (see e.g.\ Chapter~16 of~\cite{Weg05a}).
In the proof of this theorem use is made of register elimination, a 
technique similar to gate elimination as used in work on circuit lower 
bounds.

The following result is a corollary of the strengthening of 
Theorem~\ref{theorem-NZTISp-shortest} that is actually proved above.
\begin{corollary}
\label{corollary-NZTISp-shortest}
For all $n \in \Natpos$, for all $X \in \ISbr$ of the form
$\outbr.\setbr{1} \conc Y$ or $\ptst{\outbr.\setbr{1}} \conc Y$, 
$X$ computes $\tstnzc{n}$ only if $\len(X) > \len(\TSTNZp{n})$.
\end{corollary} 

\section{More shortest non-zeroness test instruction sequences}
\label{sect-ISNZ-shortest-all}

In this section, we study the remaining instruction sequences of the 
same length as $\TSTNZp{n}$ that compute $\tstnzc{n}$.
The final outcome of this study is important for the proof of 
Theorem~\ref{theorem-complexity-shortest} in 
Section~\ref{sect-complexity}.

The following proposition states that change of the order in which the 
read instructions occur in $\TSTNZp{n}$ yields again a correct 
implementation of $\tstnzc{n}$.
\begin{proposition}
\label{prop-NZTISp-perm-correct}
For each $n \in \Natpos$, 
for each bijection $\varrho$ on $\set{k \in \Natpos \where k \leq n}$, 
$\tstnzc{n}$ is also computed by the instruction sequence obtained from 
$\TSTNZp{n}$ by replacing, for each $j \in \Natpos$ with $j \leq n$, 
all occurrences of the register name $\inbr{j}$ in $\TSTNZp{n}$ by 
$\inbr{\varrho(j)}$. 
\end{proposition}
\begin{proof}
The proof is like the proof of Proposition~\ref{prop-NZTISp-correct}, 
but with, for each $j \in \Natpos$ with $j \leq n$, all occurrences of 
the register name $\inbr{j}$ in the proof replaced by 
$\inbr{\varrho(j)}$.
\qed
\end{proof}
The proof of Proposition~\ref{prop-NZTISp-correct} can be seen as a 
special case of the proof of Proposition~\ref{prop-NZTISp-perm-correct}, 
namely the case 
where $\varrho$ is the identity function.

The following proposition states that, for instruction sequences $X$ as 
considered in Proposition~\ref{prop-NZTISp-perm-correct}, in the case 
that $n$ is odd, change of the position of the isolated read instruction
of $X$ yields again a correct implementation of $\tstnzc{n}$.
\begin{proposition}
\label{prop-NZTISp-perm-like-correct}
For each odd $n \in \Natpos$, 
for each bijection $\varrho$ on $\set{k \in \Natpos \where k \leq n}$ 
and $m \in \Nat$ with $m \leq (n-1)/2$,
$\tstnzc{n}$ is also computed by the instruction sequence 
$\Conc{i = 1}{m} 
  (\ntst{\inbr{\varrho(2{\mul}i{-}1)}.\getbr} \conc
   \ptst{\inbr{\varrho(2{\mul}i)}.\getbr} \conc
   \outbr.\setbr{1}) \conc 
\linebreak[2]
 \ptst{\inbr{\varrho(2{\mul}m{+}1)}.\getbr} \conc \outbr.\setbr{1} \conc 
\linebreak[2]
 \Conc{i = m+1}{(n-1)/2} 
  (\ntst{\inbr{\varrho(2{\mul}i)}.\getbr} \conc
   \ptst{\inbr{\varrho(2{\mul}i{+}1)}.\getbr} \conc
   \outbr.\setbr{1}) \conc 
 \halt$.
\end{proposition}
\begin{proof}
\sloppy
Let $n \in \Natpos$ be odd, and
let $\varrho$ be a bijection on $\set{k \in \Natpos \where k \leq n}$.
For each $m \in \Nat$ with $m \leq (n-1)/2$, 
let 
$X_m =
 \Conc{i = 1}{m} 
  (\ntst{\inbr{\varrho(2{\mul}i{-}1)}.\getbr} \conc
   \ptst{\inbr{\varrho(2{\mul}i)}.\getbr} \conc \nolinebreak
   \outbr.\setbr{1}) \conc 
 \ptst{\inbr{\varrho(2{\mul}m{+}1)}.\getbr} \conc \outbr.\setbr{1} \conc 
 \Conc{i = m+1}{(n-1)/2} 
  (\ntst{\inbr{\varrho(2{\mul}i)}.\getbr} \conc
   \ptst{\inbr{\varrho(2{\mul}i{+}1)}.\getbr} \conc \nolinebreak
   \outbr.\setbr{1}) \conc \nolinebreak
 \halt$.
We prove that $X_m$ computes $\tstnzc{n}$ by induction on $m$. 
The basis step follows immediately from 
Proposition~\ref{prop-NZTISp-perm-correct}.
The inductive step goes as follows.
Let 
$Y =  
\ptst{\inbr{\varrho(2{\mul}m{+}1)}.\getbr} \conc \outbr.\setbr{1} \conc 
\ntst{\inbr{\varrho(2{\mul}m{+}2)}.\getbr} \conc
\ptst{\inbr{\varrho(2{\mul}m{+}3)}.\getbr} \conc \outbr.\setbr{1}$ and
$Y' =  
\ntst{\inbr{\varrho(2{\mul}m{+}1)}.\getbr} \conc  
\ptst{\inbr{\varrho(2{\mul}m{+}2)}.\getbr} \conc \outbr.\setbr{1} \conc
\ptst{\inbr{\varrho(2{\mul}m{+}3)}.\getbr} \conc \outbr.\setbr{1}$.
Then $X_{m+1}$ is $X_m$ with the subsequence $Y$ replaced by $Y'$.
From this, it follows easily by a case distinction on the contents of 
$\inbr{\varrho(2{\mul}m{+}1)}$, $\inbr{\varrho(2{\mul}m{+}2)}$, and 
$\inbr{\varrho(2{\mul}m{+}3)}$ that $X_m$ and $X_{m+1}$ compute the same 
function from $\Bool^n$ to $\Bool$.
From this and the induction hypothesis, it follows that $X_{m+1}$ 
computes $\tstnzc{n}$.
\qed
\end{proof}

Let $X$ be an instruction sequence as considered in 
Propositions~\ref{prop-NZTISp-perm-correct} 
or~\ref{prop-NZTISp-perm-like-correct}.
Then replacement of one or more occurrences of $\outbr.\setbr{1}$ in $X$ 
by $\ptst{\outbr.\setbr{1}}$ yields again a correct implementation of 
$\tstnzc{n}$ because the effects of these instructions on execution of 
$X$ are always the same.
Even replacement of one or more occurrences of $\outbr.\setbr{1}$ in 
$X$ by $\ntst{\outbr.\setbr{1}}$ yields again a correct implementation 
of $\tstnzc{n}$, unless its last occurrence is replaced, because 
checking of the first of the remaining input registers can be skipped 
once an input register is encountered whose content is $1$. 
Moreover, replacement of one or more occurrences of $\outbr.\setbr{1}$ 
in $X$ by a forward jump instruction that leads to another occurrence of 
$\outbr.\setbr{1}$ yields again a correct implementation of $\tstnzc{n}$ 
because, once an input register is encountered whose content is $1$, 
checking of the remaining input registers can be skipped.

In the preceding paragraph, the clarifying intuitions for the statements 
made do not sufficiently verify the statements.
Below, the statements are incorporated into 
Theorem~\ref{theorem-NZTISp-like-correct} and verified via the proof of 
that theorem.

We define a set $\TSTNZpc{n}$ of instruction sequences with 
Propositions~\ref{prop-NZTISp-perm-correct} 
and~\ref{prop-NZTISp-perm-like-correct} 
and the statements made above in mind. \\[2ex]
For even $n$, we define a subset $\TSTNZpc{n}$ of $\ISbr$ as follows: 
\\[2ex]
$X \in \TSTNZpc{n}$ iff 
\begin{ldispl}
X = 
\Conc{i = 1}{n/2} 
 (\ntst{\inbr{\varrho(2{\mul}i{-}1)}.\getbr} \conc
  \ptst{\inbr{\varrho(2{\mul}i)}.\getbr} \conc \varphi(i)) \conc 
\halt 
\end{ldispl}%
\hspace*{1.5em}for some function $\varphi$ from 
$\set{k \in \Natpos \where k \leq n/2}$ to $\PInstr$ such that 
\begin{ldispl}
\begin{array}[t]{@{}l@{\;}c@{\;}l}
\varphi(j) & \in & 
\set{\fjmp{3{\mul}k} \where k \in \Natpos \Land k \leq n/2 - j} 
\\ & \union &
\set{\outbr.\setbr{1},\ptst{\outbr.\setbr{1}},\ntst{\outbr.\setbr{1}}}\;, 
\\
\varphi(n/2) & \in & \set{\outbr.\setbr{1},\ptst{\outbr.\setbr{1}}}
\end{array}
\end{ldispl}%
\hspace*{1.5em}and some bijection $\varrho$ on 
$\set{k \in \Natpos \where k \leq n}$. \\[2ex]
For odd $n$, we define a subset $\TSTNZpc{n}$ of $\ISbr$ as follows: 
\\[2ex]
$X \in \TSTNZpc{n}$ iff 
\begin{ldispl}
\mathrm{there\; exists\; an\;} 
m \in \set{k \in \Nat \where k \leq (n-1)/2} \mathrm{\;such\; that} \\
\begin{array}{@{}l@{\;}c@{\;}l}
X & = &
\Conc{i = 1}{m} 
 (\ntst{\inbr{\varrho(2{\mul}i{-}1)}.\getbr} \conc
  \ptst{\inbr{\varrho(2{\mul}i)}.\getbr} \conc \varphi'_m(i)) \conc {} 
\\ & &
\ptst{\inbr{\varrho(2{\mul}m{+}1)}.\getbr} \conc 
\varphi'_m(m{+}1) \conc {} 
\\ & &
\Conc{i = m+1}{(n-1)/2} 
 (\ntst{\inbr{\varrho(2{\mul}i)}.\getbr} \conc
  \ptst{\inbr{\varrho(2{\mul}i{+}1)}.\getbr} \conc
  \varphi'_m(i{+}1)) \conc 
\halt 
\end{array}
\end{ldispl}%
\hspace*{1.5em}for some function $\varphi'_m$ from 
$\set{k \in \Natpos \where k \leq (n+1)/2}$ to $\PInstr$ such that 
\begin{ldispl}
\begin{array}{@{}l@{\;}c@{\;}l}
\varphi'_m(j) & \in &
\set{\fjmp{3{\mul}k{-}1} \where
     k \in \Natpos \Land k \leq (n+1)/2 - j \Land j \leq m < j + k} 
\\ & \union &
\set{\fjmp{3{\mul}k} \where
     k \in \Natpos \Land k \leq (n+1)/2 - j \Land 
                                            \Lnot j \leq m < j + k} 
\\ & \union &
\set{\outbr.\setbr{1},\ptst{\outbr.\setbr{1}},\ntst{\outbr.\setbr{1}}}\;, 
\\
\varphi'_m((n+1)/2) & \in & 
\set{\outbr.\setbr{1},\ptst{\outbr.\setbr{1}}}
\end{array}
\end{ldispl}%
\hspace*{1.5em}and some bijection $\varrho$ on 
$\set{k \in \Natpos \where k \leq n}$. 

Obviously, we have that $\TSTNZp{n} \in \TSTNZpc{n}$ and,
for each $X \in \TSTNZpc{n}$, $\len(X) = \len(\TSTNZp{n})$.

The following theorem states that each instruction sequence from 
$\TSTNZpc{n}$ correctly implements $\tstnzc{n}$.
\begin{theorem}
\label{theorem-NZTISp-like-correct}
For all $n \in \Natpos$, for all $X \in \ISbr$, 
$X \in \TSTNZpc{n}$ only if $X$ computes $\tstnzc{n}$.
\end{theorem}
\begin{proof}
For convenience, forward jump instructions, $\outbr.\setbr{1}$, 
$\ptst{\outbr.\setbr{1}}$, and $\ntst{\outbr.\setbr{1}}$ are called
\emph{replaceable} instructions in this proof.

Let $n \in \Natpos$, and 
let $X \in \ISbr$ be such that $X \in \TSTNZpc{n}$.
Let $Y \in \ISbr$ be obtained from $X$ by replacing all occurrences of 
a replaceable instruction other than $\outbr.\setbr{1}$ in $X$ by 
$\outbr.\setbr{1}$.
It follows immediately from Propositions~\ref{prop-NZTISp-perm-correct} 
and~\ref{prop-NZTISp-perm-like-correct} that $Y$ computes $\tstnzc{n}$. 
Hence, it remains to be proved that $X$ and $Y$ compute the same 
function from $\Bool^n$ to $\Bool$.

The fact that $X$ and $Y$ compute the same function from $\Bool^n$ to 
$\Bool$ is proved by induction on the number of occurrences of 
replaceable instructions other than $\outbr.\setbr{1}$ in $X$.
The basis step is trivial.
The inductive step goes as follows. 
Let $X' \in \ISbr$ be obtained from $X$ by replacing the first 
occurrence of a replaceable instruction other than $\outbr.\setbr{1}$ in 
$X$ by $\outbr.\setbr{1}$.
From the induction hypothesis and the fact that $Y$ computes 
$\tstnzc{n}$, it follows that $X'$ computes $\tstnzc{n}$.
Clearly, execution of $X$ and $X'$ yield the same final content of 
$\outbr$ if the initial contents of $\inbr{1},\ldots,\inbr{n}$ are 
such that execution of $X'$ does not proceed at some point with the 
replacing occurrence of $\outbr.\setbr{1}$.
What remains to be shown is that execution of $X$ and $X'$ yield the 
same final content of $\outbr$ if the initial contents of 
$\inbr{1},\ldots,\inbr{n}$ are such that execution of $X'$ proceeds at 
some point with the replacing occurrence of $\outbr.\setbr{1}$.
Call this case the decisive case.
If the decisive case occurs, then the content of at least one of the 
input registers $\inbr{1},\ldots,\inbr{n}$ is $1$.
From this and the fact that $X'$ computes $\tstnzc{n}$, it follows that 
on execution of $X'$ the final content of $\outbr$ is $1$ in the 
decisive case. 
Execution of $X$ and execution of $X'$ have the same effects in the 
decisive case until the point where $X'$ proceeds with the replacing 
occurrence of $\outbr.\setbr{1}$.
At that point, execution of $X$ proceeds with the replaced occurrence 
of a replaceable instruction other than $\outbr.\setbr{1}$ instead.
From this, the fact that $X$ contains no instructions by which the 
content of $\outbr$ can become $0$, and the fact that $X$ contains only 
forward jump instructions that lead in one or more steps to a 
replaceable instruction other than a forward jump instruction, it 
follows that on execution of $X$ the final content of $\outbr$ is also 
$1$ in the decisive case.
Hence, $X$ and $Y$ compute the same function from $\Bool^n$ to $\Bool$.
\qed
\end{proof}

There are instruction sequences in $\TSTNZpc{n}$ in which there is only
one occurrence of $\outbr.\setbr{1}$ and no occurrences of
$\ptst{\outbr.\setbr{1}}$ or $\ntst{\outbr.\setbr{1}}$.
These instruction sequences compute $\tstnzc{n}$ much more efficiently 
than $\TSTNZp{n}$ because, once an input register is encountered whose 
content is $1$, checking of the remaining input registers is skipped.

Not all instruction sequences with the same length as $\TSTNZp{n}$ that 
correctly implement $\tstnzc{n}$ belong to $\TSTNZpc{n}$.
If $n$ is odd and $n > 1$, $\fjmp{2}$ may occur once in an instruction 
sequence that belongs to $\TSTNZpc{n}$.
Let $X \in \TSTNZpc{n}$ be such that $\fjmp{2}$ occurs in $X$, and
let $i \in \Natpos$ be such that $\ptst{\inbr{i}.\getbr}$ or 
$\ntst{\inbr{i}.\getbr}$ occurs before $\fjmp{2}$ in $X$.
If the occurrence of $\fjmp{2}$ will only be executed if the content of 
$\inbr{i}$ is $0$, then its replacement by an occurrence of 
$\ptst{\inbr{i}.\getbr}$ yields a correct implementation of $\tstnzc{n}$
that does not belong to $\TSTNZpc{n}$.
Because of this, we introduce an extension of $\TSTNZpc{n}$.

We define a subset $\TSTNZpce{n}$ of $\ISbr$ as follows: 
\begin{itemize}
\item
if no input register is read more than once in $X$, then 
$X \in \TSTNZpce{n}$ iff $X \in \TSTNZpc{n}$;
\item
if at least two input registers are read more than once in $X$ or at 
least one input register is read more than twice in $X$, then 
$X \notin \TSTNZpce{n}$;
\item
if exactly one input register is read more than once in $X$ and the 
input register concerned is read exactly twice in $X$, then 
$X \in \TSTNZpce{n}$ iff
\begin{itemize}
\item[(a)]
no chain of zero or more forward jump instructions in $X$ that begins 
at or before its $(k+2)$th primitive instruction leads to its $l$th 
primitive instruction of $X$,
\item[(b)]
the $l$th primitive instruction of $X$ is $\ptst{\inbr{i}.\getbr}$ for 
some $i \leq n$,
\item[(c)]
$X'\in \TSTNZpc{n}$, 
\end{itemize}
where $k$ and $l$, with $k \neq l$, are such that the same input 
register is read by the $k$th and $l$th primitive instruction of $X$, 
and where $X'$ is $X$ with the $l$th primitive instruction of $X$ 
replaced by $\fjmp{2}$.
\end{itemize}

Obviously, we have that $\TSTNZp{n} \in \TSTNZpce{n}$, and
for each $X \in \TSTNZpce{n}$, $\len(X) = \len(\TSTNZp{n})$.
Moreover, we have that $\TSTNZpc{n} = \TSTNZpce{n}$ if $n$ is even and
$\TSTNZpc{n} \neq \TSTNZpce{n}$ if $n$ is odd and $n > 3$.

The following theorem states that each instruction sequence from 
$\TSTNZpce{n}$ correctly implements $\tstnzc{n}$.
\begin{theorem}
\label{theorem-more-NZTISp-like-correct}
For all $n \in \Natpos$, for all $X \in \ISbr$, 
$X \in \TSTNZpce{n}$ only if $X$ computes $\tstnzc{n}$.
\end{theorem}
\begin{proof}
By Theorem~\ref{theorem-NZTISp-like-correct}, it suffices to show that 
$X$ computes $\tstnzc{n}$ if exactly one input register is read more 
than once in $X$, that input register is read exactly twice in $X$, and 
$X$ satisfies conditions~(a), (b), and~(c) from the definition of 
$\TSTNZpce{n}$.

We know from the definition of $\TSTNZpc{n}$, that $\fjmp{2}$ occurs at 
most once in an instruction sequence from $\TSTNZpc{n}$.
By conditions~(b) and~(c), there is a $Y \in \TSTNZpc{n}$ in which 
$\fjmp{2}$ occurs such that $X$ is $Y$ with the (unique) occurrence of
$\fjmp{2}$ replaced by $\ptst{\inbr{i}.\getbr}$ for some $i \leq n$.
Let $k$ and $l$, with $k \neq l$, be such that the same input register 
is read by the $k$th and $l$th primitive instruction of $X$
and the $l$th primitive instruction of $X$ is the replacing instruction.

If $k > l$ and the initial content of the register read by the replacing 
instruction is $0$, its execution has the same effect as the execution 
of $\fjmp{2}$ and consequently $X$ computes $\tstnzc{n}$.
If $k > l$ and the initial content of the register read by the replacing 
instruction is $1$, although its execution has the same effect as the 
execution of $\fjmp{1}$, the $k$th primitive instruction is eventually 
executed and consequently $X$ computes $\tstnzc{n}$.

If $k < l$, by Theorem~\ref{theorem-NZTISp-like-correct} and the 
definition of $\TSTNZpc{n}$, $X$ computes $\tstnzc{n}$ iff this 
replacing instruction is executed only in the case that its execution 
has the same effect as the execution of $\fjmp{2}$ or its execution is 
preceded by the execution of an occurrence of $\outbr.\setbr{1}$, 
$\ptst{\outbr.\setbr{1}}$ or $\ntst{\outbr.\setbr{1}}$.
So, what is left to be shown is the following claim: 
\begin{quote}
if $k < l$, the replacing 
instruction is  executed only in the case that the initial content of 
the input register involved is $0$ or its execution is preceded by the 
execution of an occurrence of $\outbr.\setbr{1}$, 
$\ptst{\outbr.\setbr{1}}$ or $\ntst{\outbr.\setbr{1}}$.
\end{quote}
By the definition of $\TSTNZpc{n}$, this claim is right if the following 
conditions are satisfied:
\begin{pplist}
\item[(1)] no chain of forward jump instructions that begins before the $k$th 
primitive instruction of $X$ leads to the $l$th primitive instruction of 
$X$;
\item[(2)] the chain of (zero or more) forward jump instructions that begins 
at the $(k+1)$th primitive instruction of $X$ if the $(k+1)$th primitive 
instruction of $X$ is not a read instruction and at the $(k+2)$th 
primitive in\-struc\-tion of $X$ otherwise does not lead to the $l$th 
primitive instruction of $X$.
\end{pplist}
We know from the definition of $\TSTNZpce{n}$ that the $(k+1)$th 
primitive instruction of $X$ is not a read instruction iff the $(k+2)$th 
primitive instruction of $X$ is a read instruction.
Therefore, conditions (1) and (2) are satisfied if condition~(a) from 
the definition  of $\TSTNZpce{n}$ is satisfied.
Hence, the claim is right, and the proof is complete.
\qed
\end{proof}

The following theorem states that each instruction sequence with the 
same length as $\TSTNZp{n}$ that correctly implements $\tstnzc{n}$ 
belongs to $\TSTNZpce{n}$.
\begin{theorem}
\label{theorem-correct-NZTISp-like}
For all $n \in \Natpos$, 
for all $X \in \ISbr$ with $\len(X) = \len(\TSTNZp{n})$, 
$X$ computes $\tstnzc{n}$ only if $X \in \TSTNZpce{n}$.
\end{theorem}
\begin{proof}
We prove this theorem by induction on $n$.
In the proof, we use the notation $\chi_i(X)$, where $i \in \Natpos$,
defined at the beginning of the proof of
Theorem~\ref{theorem-NZTISp-shortest}.

The basis step consists of proving that for all $X \in \ISbr$ with 
$\len(X) = 3$, $X$ computes $\tstnzc{1}$ only if $X \in \TSTNZpce{1}$.
This follows trivially from the observations made at the beginning of 
the proof of Theorem~\ref{theorem-NZTISp-shortest}.

\sloppy
The inductive step goes in the following way. 
Suppose that $X \in \ISbr$, $\len(X) = \len(\TSTNZp{n+1})$, and $X$ 
computes $\tstnzc{n+1}$.
Without loss of generality, we may assume that 
$X = u_1 \conc \ldots \conc u_k$, where
$k = \len(\TSTNZp{n+1})$, $u_1,\ldots,u_k \in \PInstr$, and 
$u_1$ is $\ptst{\inbr{i}.\getbr}$ or $\ntst{\inbr{i}.\getbr}$ for some
$i \in \Natpos$ such that $i \leq n + 1$.
This can be seen as follows:
\begin{plist}
\item
if $u_1$ is $\halt$ or $\fjmp{l}$ with $l = 0$ or $l \geq k$, then 
$X$ cannot compute $\tstnzc{n+1}$;
\item
if $u_1$ is $\fjmp{l}$ with $0 < l < k$, then there is an 
$X' \in \ISbr$ that computes $\tstnzc{n+1}$ such that 
$\len(X') < \len(\TSTNZp{n+1})$ --- which contradicts 
Theorem~\ref{theorem-NZTISp-shortest};
\item
if $u_1$ is $\outbr.\setbr{0}$, $\ptst{\outbr.\setbr{0}}$ or 
$\ntst{\outbr.\setbr{0}}$, then $u_1$ can be replaced by $\fjmp{1}$ or 
$\fjmp{2}$ in $X$ and so there is an $X' \in \ISbr$ that computes 
$\tstnzc{n+1}$ such that $\len(X') < \len(\TSTNZp{n+1})$ --- which 
contradicts Theorem~\ref{theorem-NZTISp-shortest};
\item
if $u_1$ is $\outbr.\setbr{1}$ or $\ptst{\outbr.\setbr{1}}$, then, 
by Corollary~\ref{corollary-NZTISp-shortest}, 
$\len(X) > \len(\TSTNZp{n+1})$ --- which contradicts the assumption that
$\len(X) = \len(\TSTNZp{n+1})$;
\item
if $u_1$ is $\ntst{\outbr.\setbr{1}}$, then $u_1 \conc u_2$ can be 
replaced by $\outbr.\setbr{1}$ in $X$ and so there is an $X' \in \ISbr$ 
that computes $\tstnzc{n+1}$ such that $\len(X') < \len(\TSTNZp{n+1})$ 
--- which contradicts Theorem~\ref{theorem-NZTISp-shortest};
\item
if $u_1$ is $\auxbr{j}.\setbr{0}$, $\ptst{\auxbr{j}.\setbr{0}}$, 
$\ntst{\auxbr{j}.\setbr{0}}$, $\auxbr{j}.\getbr$, 
$\ptst{\auxbr{j}.\getbr}$ or $\ntst{\auxbr{j}.\getbr}$ for some 
$j \in \Natpos$, then $u_1$ can be replaced by $\fjmp{1}$ or $\fjmp{2}$ 
in $X$ and so there is an $X' \in \ISbr$ that computes $\tstnzc{n+1}$ 
such that $\len(X') < \len(\TSTNZp{n+1})$ --- which contradicts 
Theorem~\ref{theorem-NZTISp-shortest};
\item
if $u_1$ is $\auxbr{j}.\setbr{1}$, $\ptst{\auxbr{j}.\setbr{1}}$ or 
$\ntst{\auxbr{j}.\setbr{1}}$ for some $j \in \Natpos$, then 
$\chi_j(u_1)$ can be replaced by $\fjmp{1}$ or $\fjmp{2}$ in $\chi_j(X)$ 
and so, because $\chi_j(X)$ also computes $\tstnzc{n+1}$ and 
$\len(\chi_j(X)) = \len(X)$, there is an $X' \in \ISbr$ such that $X'$ 
computes $\tstnzc{n+1}$ and $\len(X') < \len(\TSTNZp{n+1})$ --- which 
contradicts Theorem~\ref{theorem-NZTISp-shortest};
\item
if $u_1$ is $\inbr{j}.\getbr$ for some $j \in \Natpos$, then $u_1$ can 
be replaced by $\fjmp{1}$ in $X$ and so there is an $X' \in \ISbr$ that 
computes $\tstnzc{n+1}$ such that $\len(X') < \len(\TSTNZp{n+1})$ --- 
which contradicts Theorem~\ref{theorem-NZTISp-shortest};
\item
if $u_1$ is $\ptst{\inbr{j}.\getbr}$ or $\ntst{\inbr{j}.\getbr}$ for 
some $j \in \Natpos$ such that $j > n + 1$, then, because the final 
content of $\outbr$ is independent of the initial content of $\inbr{j}$,
$u_1$ can be replaced by $\fjmp{1}$ and $\fjmp{2}$ in $X$ and so there 
is an $X' \in \ISbr$ that computes $\tstnzc{n+1}$ such that 
$\len(X') < \len(\TSTNZp{n+1})$ ---  which contradicts 
Theorem~\ref{theorem-NZTISp-shortest}.
\end{plist}
So, we distinguish between 
the case that $u_1$ is $\ptst{\inbr{i}.\getbr}$ and 
the case that $u_1$ is $\ntst{\inbr{i}.\getbr}$.
In both cases, we further distinguish between the case that $n$ is even
and the case that $n$ is odd.

The case that $u_1$ is $\ptst{\inbr{i}.\getbr}$ and $n$ is even goes as
follows.
Let $Y = u_3 \conc\nolinebreak \ldots \conc\nolinebreak u_k$, and 
let $Y' = Y[\inbr{i} = 0]$. 
Then $Y'$ computes $\tstnzc{n}$.
Moreover, by Proposition~\ref{prop-NZTISp-length}, 
$\len(Y') = \len(\TSTNZp{n})$.
Hence, by the induction hypothesis, $Y' \in \TSTNZpce{n}$.
From this and the fact that $n$ is even, it follows that $\fjmp{1}$ and 
$\fjmp{2}$ do not occur in $Y'$.
Consequently, $Y$ is simply $Y'$ with, for each $j \in \Natpos$ with 
$i < j \leq n + 1$, all occurrences of the register name $\inbr{j{-}1}$ 
replaced by $\inbr{j}$.
Moreover, because $X$ computes $\tstnzc{n+1}$, $u_2$ is 
$\outbr.\setbr{1}$, $\ptst{\outbr.\setbr{1}}$, $\ntst{\outbr.\setbr{1}}$ 
or a forward jump instruction that leads to an occurrence of such an 
instruction.
Hence, $X \in \TSTNZpce{n+1}$.

The case that $u_1$ is $\ptst{\inbr{i}.\getbr}$ and $n$ is odd goes as
follows.
Let $Y = u_3 \conc \ldots \conc u_k$, and let $Y' = Y[\inbr{i} = 0]$. 
Then $Y'$ computes $\tstnzc{n}$.
Moreover, by Proposition~\ref{prop-NZTISp-length}, 
$\len(Y') = \len(\TSTNZp{n}) - 1$.
This contradicts Theorem~\ref{theorem-NZTISp-shortest}.
Hence, this case cannot occur.

The case that $u_1$ is $\ntst{\inbr{i}.\getbr}$ and $n$ is odd goes as
follows.
Let $Y = u_2 \conc \ldots \conc u_k$, and let $Y' = Y[\inbr{i} = 0]$. 
Then $Y'$ computes $\tstnzc{n}$.
Moreover, by Proposition~\ref{prop-NZTISp-length}, 
$\len(Y') = \len(\TSTNZp{n})$.
Hence, by the induction hypothesis, $Y' \in \TSTNZpce{n}$.
From this and the fact that $n$ is odd, it follows that $\fjmp{1}$ does 
not occur in $Y'$ and $\fjmp{2}$ may occur at most once in $Y'$.
If $\fjmp{2}$ occurs in $Y'$, then it occurs immediately before the 
unique isolated read instruction of $Y'$.
However, because $X$ computes $\tstnzc{n+1}$ and 
$\len(X) = \len(\TSTNZp{n+1})$, $u_2$ must be a positive read 
instruction. 
Therefore, $u_2$ is the unique isolated read instruction of $Y'$.
From this, it follows that $\fjmp{2}$ does not occur in $Y'$.
Consequently, $Y$ is $Y'$ with, for each $j \in \Natpos$ with 
$i < j \leq n + 1$, all occurrences of the register name $\inbr{j{-}1}$ 
replaced by $\inbr{j}$.
Hence, $X \in \TSTNZpce{n+1}$.

The case that $u_1$ is $\ntst{\inbr{i}.\getbr}$ and $n$ is even goes as 
follows. 
Let $Y = u_2 \conc\nolinebreak \ldots \conc\nolinebreak u_k$, and 
let $Y' = Y[\inbr{i} = 0]$.
Then $Y'$ computes $\tstnzc{n}$. 
We cannot derive a contradiction immediately as in the previous three 
cases.
A case distinction on $u_2$ is needed.
All cases other than the case that $u_2$ is $\ptst{\inbr{j}.\getbr}$ 
and the case that $u_2$ is $\ntst{\inbr{j}.\getbr}$, for some 
$j \in \Natpos$ such that $i \neq j$ and $j \leq n + 1$, cannot occur 
because a contradiction can be derived.
This can be seen as follows:
\begin{plist}
\item
if $u_2$ is $\halt$ or $\fjmp{l}$ with $l = 0$ or $l \geq k - 1$, then 
$Y'$ cannot compute $\tstnzc{n}$;
\item
if $u_2$ is $\fjmp{1}$ or $\fjmp{2}$, then, because 
$\ntst{\inbr{i}.\getbr} \conc \fjmp{1}$ can be replaced by 
$\inbr{i}.\getbr$ in $X$ and 
$\ntst{\inbr{i}.\getbr} \conc \fjmp{2}$ can be replaced by 
$\ptst{\inbr{i}.\getbr}$ in $X$, 
there is an $X' \in \ISbr$ that computes 
$\tstnzc{n+1}$ such that $\len(X') < \len(\TSTNZp{n+1})$ --- which 
contradicts Theorem~\ref{theorem-NZTISp-shortest};
\item
if $u_2$ is $\fjmp{l}$ with $2 < l < k - 1$, then 
there is an $Y'' \in \ISbr$ such that $Y''$ computes $\tstnzc{n}$ and, 
by Proposition~\ref{prop-NZTISp-length}, $\len(Y'') < \len(\TSTNZp{n})$ 
--- which contradicts 
Theorem~\ref{theorem-NZTISp-shortest};
\item
if $u_2$ is $\outbr.\setbr{0}$, $\ptst{\outbr.\setbr{0}}$ or 
$\ntst{\outbr.\setbr{0}}$, then $u_2$ can be replaced by $\fjmp{1}$ or 
$\fjmp{2}$ in $X$ and so, because 
$\ntst{\inbr{i}.\getbr} \conc \fjmp{1}$ can be replaced by 
$\inbr{i}.\getbr$ and 
$\ntst{\inbr{i}.\getbr} \conc \fjmp{2}$ can be replaced by 
$\ptst{\inbr{i}.\getbr}$, 
there is an $X' \in \ISbr$ that computes 
$\tstnzc{n+1}$ such that $\len(X') < \len(\TSTNZp{n+1})$ --- which 
contradicts Theorem~\ref{theorem-NZTISp-shortest};  
\item
if $u_2$ is $\outbr.\setbr{1}$ or $\ptst{\outbr.\setbr{1}}$, then we
make a case distinction on the initial content of $\inbr{i}$:
(a)~if $\inbr{i}$ contains $0$, then, because $X$ computes 
$\tstnzc{n+1}$, execution of $Y$ yields the right final content of 
$\outbr$;
(b)~if $\inbr{i}$ contains $1$, then, because $X$ computes 
$\tstnzc{n+1}$, execution of $u_3 \conc \ldots \conc u_k$ always yields 
$1$ as the final content of $\outbr$ and consequently, because there 
exist no instructions to read the content of $\outbr$, $Y$ always yields 
$1$ as the final content of $\outbr$ too;
hence $Y$ computes $\tstnzc{n+1}$ and so there is an $X' \in \ISbr$ that computes $\tstnzc{n+1}$ such that 
$\len(X') < \len(\TSTNZp{n+1})$ --- which contradicts 
Theorem~\ref{theorem-NZTISp-shortest};
\item
if $u_2$ is $\ntst{\outbr.\setbr{1}}$, then $u_2 \conc u_3$ can be 
replaced by $\outbr.\setbr{1}$ in $Y'$ and so, by 
Proposition~\ref{prop-NZTISp-length}, there is a $Y'' \in \ISbr$ of the 
form $\outbr.\setbr{1} \conc Z$ such that $Y''$ computes $\tstnzc{n}$ 
and $\len(Y'') = \len(\TSTNZp{n})$ and consequently, by the induction 
hypothesis, $Y'' \in \TSTNZpce{n}$ --- which contradicts the fact that 
$Y''$ is of the form $\outbr.\setbr{1} \conc Z$;
\item
\sloppy
if $u_2$ is $\auxbr{j'}.\setbr{0}$, $\ptst{\auxbr{j'}.\setbr{0}}$, 
$\ntst{\auxbr{j'}.\setbr{0}}$, $\auxbr{j'}.\getbr$, 
$\ptst{\auxbr{j'}.\getbr}$ or $\ntst{\auxbr{j'}.\getbr}$ for some 
$j' \in \Natpos$, then $u_2$ can be replaced by $\fjmp{1}$ or $\fjmp{2}$ 
in $X$ and so, because 
$\ntst{\inbr{i}.\getbr} \conc \fjmp{1}$ can be replaced by 
$\inbr{i}.\getbr$ and 
$\ntst{\inbr{i}.\getbr} \conc \fjmp{2}$ can be replaced by 
$\ptst{\inbr{i}.\getbr}$, 
there is an $X' \in \ISbr$ that computes $\tstnzc{n+1}$ 
such that $\len(X') < \len(\TSTNZp{n+1})$ --- which contradicts 
Theorem~\ref{theorem-NZTISp-shortest};
\item
if $u_2$ is $\auxbr{j'}.\setbr{1}$ or $\ptst{\auxbr{j'}.\setbr{1}}$ for 
some $j' \in \Natpos$, then $\chi_{j'}(u_2)$ can be replaced by 
$\fjmp{1}$ in $\chi_{j'}(Y')$ and so, by 
Proposition~\ref{prop-NZTISp-length}, there is an $Y'' \in \ISbr$ that 
computes $\tstnzc{n}$ such that $\len(Y'') = \len(\TSTNZp{n})$ and 
consequently, by the induction hypothesis, $Y'' \in \TSTNZpce{n}$;
therefore an instruction of the form $\auxbr{j'}.\getbr$, 
$\ptst{\auxbr{j'}.\getbr}$ or $\ntst{\auxbr{j'}.\getbr}$ does not
occur in $Y$ and so $u_2$ can be replaced by $\fjmp{1}$ in $X$ and, 
because $\ntst{\inbr{i}.\getbr} \conc \fjmp{1}$ can be replaced by
$\ptst{\inbr{i}.\getbr}$, there is an $X' \in \ISbr$ that computes 
$\tstnzc{n+1}$ such that $\len(X') < \len(\TSTNZp{n+1})$ --- which 
contradicts Theorem~\ref{theorem-NZTISp-shortest};
\item
if $u_2$ is $\ntst{\auxbr{j'}.\setbr{1}}$ for some $j' \in \Natpos$, 
then $u_2 \conc u_3$ can be replaced by $\auxbr{j'}.\setbr{1}$ in 
$Y'$ and so, by Proposition~\ref{prop-NZTISp-length}, there is a 
$Y'' \in \ISbr$ of the form $\auxbr{j'}.\setbr{1} \conc Z$ such 
that $Y''$ computes $\tstnzc{n}$ and $\len(Y'') = \len(\TSTNZp{n})$ and 
consequently, by the induction hypothesis, $Y'' \in \TSTNZpce{n}$ --- 
which contradicts the fact that $Y''$ is of the form 
$\auxbr{j'}.\setbr{1} \conc Z$;
\item
if $u_2$ is $\inbr{j'}.\getbr$ for some $j' \in \Natpos$, then 
$\inbr{j'}.\getbr$ can be replaced by $\fjmp{1}$ in $X$ and so, because 
$\ntst{\inbr{i}.\getbr} \conc \fjmp{1}$ can be replaced by 
$\inbr{i}.\getbr$, there is an $X' \in \ISbr$ that computes 
$\tstnzc{n+1}$ such that $\len(X') < \len(\TSTNZp{n+1})$ --- which 
contradicts Theorem~\ref{theorem-NZTISp-shortest};
\item
if $u_2$ is $\ptst{\inbr{j'}.\getbr}$ or $\ntst{\inbr{j'}.\getbr}$ for 
some $j' \in \Natpos$ such that $j' > n + 1$, then, because the final 
content of $\outbr$ is independent of the initial content of 
$\inbr{j'}$, $u_2$ can be replaced by $\fjmp{1}$ and $\fjmp{2}$ in $X$ 
and so, because $\ntst{\inbr{j}.\getbr} \conc \fjmp{1}$ can be replaced 
by $\inbr{j}.\getbr$ and $\ntst{\inbr{j}.\getbr} \conc \fjmp{2}$ can be 
replaced by $\ptst{\inbr{j}.\getbr}$, there is an $X' \in \ISbr$ that 
computes $\tstnzc{n+1}$ such that $\len(X') < \len(\TSTNZp{n+1})$ --- 
which contradicts Theorem~\ref{theorem-NZTISp-shortest};
\item
if $u_2$ is $\ptst{\inbr{i}.\getbr}$ or $\ntst{\inbr{i}.\getbr}$, then, 
because 
$\ntst{\inbr{i}.\getbr} \conc \ptst{\inbr{i}.\getbr}$ can be replaced by 
$\ptst{\inbr{i}.\getbr}$ in $X$ and
$\ntst{\inbr{i}.\getbr} \conc \ntst{\inbr{i}.\getbr}$ can be replaced by 
$\inbr{i}.\getbr$ in $X$,
there is an $X' \in \ISbr$ that computes $\tstnzc{n+1}$ such that 
$\len(X') < \len(\TSTNZp{n+1})$ --- which contradicts 
Theorem~\ref{theorem-NZTISp-shortest}.   
\end{plist}
So, we distinguish between 
the case that $u_2$ is $\ptst{\inbr{j}.\getbr}$ and 
the case that $u_2$ is $\ntst{\inbr{j}.\getbr}$.

The case that $u_2$ is $\ntst{\inbr{j}.\getbr}$ goes as follows.
Let $Y = u_3 \conc \ldots \conc u_k$, and let $Y' = Y[\inbr{j} = 0]$. 
Then, because $X$ computes $\tstnzc{n+1}$, we have:
(i)~if the initial content of $\inbr{i}$ is $1$, execution of $Y$ 
yields $1$ as final content of $\outbr$,
(ii)~if the initial content of $\inbr{i}$ is $0$, execution of $Y'$ 
yields $1$ as final content of $\outbr$ iff the initial content of at 
least one of the input registers $\inbr{1},\ldots,\inbr{n{+}1}$ other 
than $\inbr{j}$ is $1$.
From~(i), it follows that, if the initial content of $\inbr{i}$ is $1$, 
execution of $Y'$ yields $1$ as final content of $\outbr$. 
From this and~(ii), it follows that, $Y'$ computes $\tstnzc{n}$.
Moreover, by Proposition~\ref{prop-NZTISp-length}, 
$\len(Y') = \len(\TSTNZp{n})$.
Hence, by the induction hypothesis, $Y' \in \TSTNZpce{n}$.
Consequently, knowing that $n$ is even, we have that
(a)~$\inbr{j}$ is not read by a primitive instruction occurring in 
$u_3 \conc \ldots \conc u_k$ --- because otherwise there are occurrences 
of $\fjmp{1}$ or $\fjmp{2}$ in $Y'$, which contradicts 
$Y' \in \TSTNZpce{n}$ --- and
(b)~$u_3$ is $\ntst{\inbr{j'}.\getbr}$ and 
$u_4$ is $\ptst{\inbr{j''}.\getbr}$, for $j'$ and $j''$ such that 
$j' \neq j''$.
From this, it follows that execution of $X$ yields $0$ as final content 
of $\outbr$ if the initial content of $\inbr{j}$ is $1$ and the initial 
content of all other $n$ input registers is $0$.
This contradicts the fact that $X$ computes $\tstnzc{n+1}$.
Hence, also this case cannot occur.
 
The case that $u_2$ is $\ptst{\inbr{j}.\getbr}$ goes as follows.
Let $Y = u_4 \conc \ldots \conc u_k$, and 
let $Y' = Y[\inbr{i} = 0][\inbr{j} = 0]$. 
Then $Y'$ computes $\tstnzc{n-1}$.
Moreover, by Corollary~\ref{corollary-NZTISp-length}, 
$\len(Y') = \len(\TSTNZp{n-1})$.
Hence, by the induction hypothesis, $Y' \in \TSTNZpce{n-1}$.
From this and the fact that $n-1$ is odd, it follows that $\fjmp{1}$ 
does not occur in $Y'$ and $\fjmp{2}$ may occur at most once in $Y'$.
If $\fjmp{2}$ does not occur in $Y'$, then $Y = Y'$.
If $\fjmp{2}$ occurs in $Y'$, then it occurs immediately before the 
unique isolated read instruction of $Y'$.
Moreover, if $\fjmp{2}$ occurs in $Y'$, then it replaces an occurrence 
of $\ptst{\inbr{i}.\getbr}$ or $\ptst{\inbr{j}.\getbr}$ in $Y$.
Consequently, if $\fjmp{2}$ occurs in $Y'$, $Y$ is $Y'$ with either, 
for each $j' \in \Natpos$ with $i < j' \leq n + 1$, all occurrences of 
the register name $\inbr{j'{-}1}$ replaced by $\inbr{j'}$ and the single 
occurrence of $\fjmp{2}$ in $Y'$ replaced by $\ptst{\inbr{i}.\getbr}$ 
or, 
for each $j' \in \Natpos$ with $j < j' \leq n + 1$, all occurrences of 
the register name $\inbr{j'{-}1}$ replaced by $\inbr{j'}$ and the single 
occurrence of $\fjmp{2}$ in $Y'$ replaced by  $\ptst{\inbr{j}.\getbr}$.
Moreover, because $X$ computes $\tstnzc{n+1}$ and 
$\len(X) = \len(\TSTNZp{n+1})$, $u_3$ must be $\outbr.\setbr{1}$, 
$\ptst{\outbr.\setbr{1}}$, $\ntst{\outbr.\setbr{1}}$ or a forward jump 
instruction that leads to an occurrence of such an instruction.
Hence, $X \in \TSTNZpce{n+1}$.
\qed
\end{proof}

The following corollary of 
Theorems~\ref{theorem-more-NZTISp-like-correct} 
and~\ref{theorem-correct-NZTISp-like} is used in 
Section~\ref{sect-complexity}.
\begin{corollary}
\label{corollary-correct-NZTISp-like}
For all $n \in \Natpos$, 
for all $X \in \ISbr$ with $\len(X) = \len(\TSTNZp{n})$,
$X$ computes $\tstnzc{n}$ iff $X \in \TSTNZpce{n}$.
\end{corollary}

The following corollary of Theorem~\ref{theorem-NZTISp-shortest} and
Corollary~\ref{corollary-correct-NZTISp-like} is interesting because it
tells us that the use of auxiliary Boolean registers is lacking in all 
shortest instruction sequences computing $\tstnzc{n}$ 
(cf.~\cite{BM14e}).
\begin{corollary}
\label{corollary-correct-aux}
For all $n \in \Natpos$, for all $X \in \ISbr$ that compute 
$\tstnzc{n}$, $\len(X) > \len(\TSTNZp{n})$ if a register name of the 
form $\auxbr{i}$ appears in a primitive instruction occurring in $X$.
\end{corollary}

\section{The Complexity of the Correctness Problem}
\label{sect-complexity}

In this section, we study the time complexity of several restrictions of 
the problem of deciding whether an arbitrary instruction sequence from 
$\ISbr$ correctly implements the function $\tstnzc{n}$, for $n > 0$.
The restrictions considered, with the exception of one, concern only the 
length of the arbitrary instruction sequence.
The model of computation used for time complexities in this section is 
the random access machine (RAM) as described in~\cite{CR73a}.
This is made explicit in the formulation of theorems and lemmas by the 
phrase \emph{on a RAM}, but left implicit in their proofs.

A RAM consists of a read-only input tape, a write-only output tape, a 
memory consisting of an unbounded number of direct and indirect 
addressable registers that can contain an arbitrary integer, and a 
program.
The program for a RAM is a sequence of instructions, where the 
instructions include input/output instructions, arithmetic instructions,
copy instructions, and jump instructions.
The RAM model of computation described in~\cite{CR73a} differs slightly
from the one described in~\cite[Chapter~1]{AHU74a}.
The main difference is that multiplication and division instructions are
absent in the former model and present in the latter model.

Time complexities on a RAM are usually under one of the following two 
cost criteria: the uniform cost criterion and the logarithmic cost 
criterion (terminology from~\cite[Chapter~1]{AHU74a}).
The time complexities mentioned in this section are time complexities 
under the uniform cost criterion.
It is a well-known fact that, if a problem can be solved in $O(T(n))$ 
time on a RAM (without multiplication and division instructions) under 
the uniform cost criterion, then it can be solved in $O(T^3(n))$ time on 
a multi-tape Turing machine.

Most primitive instructions that may occur in instruction sequences 
from $\ISbr$ can be looked upon as consisting of two parts: a form of
instruction and a natural number.
In the case that the RAM model of computation described in~\cite{CR73a}
is used, it is contributive to the efficiency of algorithms to represent 
each instruction by two integers: one representing a form of instruction
and the other being a natural number if that is needed for the form of 
instruction concerned and $-1$ otherwise. 
Therefore, the time complexities mentioned in this section are based on 
such a representation of primitive instructions.

Firstly, we consider the problem of determining whether an arbitrary 
instruction sequence from $\ISbr$ whose length is $\len(\TSTNZp{n})$ 
correctly implements the function $\tstnzc{n}$, for $n > 0$. 
The following theorem states that this problem can be solved in $O(n^2)$
time on a RAM.
\begin{theorem}
\label{theorem-complexity-shortest}
The problem of deciding whether an $X \in \ISbr$ such that 
$\len(X) = \len(\TSTNZp{n})$ computes $\tstnzc{n}$, for $n \in \Natpos$, 
can be solved in $O(n^2)$ time on a RAM.
\end{theorem}
\begin{proof}
Let $n \in \Natpos$, and 
let $X \in \ISbr$ be such that $\len(X) = \len(\TSTNZp{n})$.
By Corollary~\ref{corollary-correct-NZTISp-like}, $X$ computes 
$\tstnzc{n}$ iff $X \in \TSTNZpce{n}$. 
That is why we alternatively prove that the membership problem for 
$\TSTNZpce{n}$ can be solved in $O(n^2)$ time.

We start with proving that the membership problem for $\TSTNZpc{n}$ can 
be solved in $O(n^2)$ time.
The definition of $\TSTNZpc{n}$ shows that the members of $\TSTNZpc{n}$ 
have a common pattern of primitive instructions which can be described 
by a regular grammar.
This pattern allows $\TSTNZpc{n}$ to be characterized as follows: 
$X \in \TSTNZpc{n}$ iff 
(i)~$X$ has this pattern,
(ii)~each of the input register names $\inbr{1},\ldots,\inbr{n}$ 
appears in one occurrence of a read instruction in $X$, 
(iii)~each occurrence of a forward jump instruction in $X$ leads, 
possibly via other forward jump instructions, to an occurrence of 
$\outbr.\setbr{1}$, $\ptst{\outbr.\setbr{1}}$ or 
$\ntst{\outbr.\setbr{1}}$, and 
(iv)~$\len(X) = \len(\TSTNZp{n})$.
Because this theorem concerns only $X \in \ISbr$ with 
$\len(X) = \len(\TSTNZp{n})$, (iv) does not have to be checked.
Checking (i), (ii), and (iii) in a straightforward way takes $O(n)$ 
steps, $O(n^2)$ steps, and $O(n^2)$ steps, respectively.
Hence, the membership problem for $\TSTNZpc{n}$, for $n \in \Natpos$, 
can be solved in $O(n^2)$ time.

We go on with proving that the membership problem for $\TSTNZpce{n}$ 
can also be solved in $O(n^2)$ time.
The membership problem for $\TSTNZpce{n}$, for $n \in \Natpos$, 
can be solved by first checking whether $X \in \TSTNZpc{n}$ and then, 
if the answer is negative, checking whether exactly one input register 
is read more than once in $X$ and the input register concerned is read 
exactly twice in $X$, determining the $k$ and $l$ such that $k$th and 
$l$th primitive instruction of $X$ are the ones that read the same input
register, and checking whether the conditions~(a), (b), and~(c) from 
the definition of $\TSTNZpce{n}$ are satisfied.
Checking whether exactly one input register is read more than once in 
$X$ and the input register concerned is read exactly twice in $X$ takes 
$O(n^2)$ steps, determining the $k$ and $l$ takes $O(n)$ steps, and 
checking whether conditions~(a), (b), and~(c) from the definition of 
$\TSTNZpce{n}$ are satisfied takes $O(n^2)$ steps, $O(n)$ steps, and 
$O(n^2)$ steps, respectively.
Hence, the membership problem for $\TSTNZpce{n}$, for $n \in \Natpos$, 
can be solved in $O(n^2)$ time.
\qed
\end{proof}

\sloppy
Secondly, we consider the problem of determining whether an arbitrary 
instruction sequence from $\ISbr$ whose length is $\len(\TSTNZp{n})$ 
plus a constant amount $m$ correctly implements the function 
$\tstnzc{n}$, for $n > 0$ and $m > 0$.
The following theorem states that this problem can be solved in 
$O((n + m) \mul 2^n)$ time on a RAM.
\begin{theorem}
\label{theorem-complexity-shortest-const-1}
The problem of deciding whether an $X \in \ISbr$ such that 
$\len(X) = \len(\TSTNZp{n}) + m$ computes $\tstnzc{n}$, for 
$n,m \in \Natpos$, can be solved in $O((n + m) \mul 2^n)$ time on a RAM.
\end{theorem}
\begin{proof}
This problem can be solved by trying out whether execution of $X$ 
yields the right final content of $\outbr$ for all $2^n$ possible 
combinations of the initial contents of $\inbr{1},\ldots,\inbr{n}$.
For each of these combinations, the trial takes $O(n + m)$ steps.
Hence, the problem of deciding whether an $X \in \ISbr$ such that 
$\len(X) = \len(\TSTNZp{n}) + m$ computes $\tstnzc{n}$, for 
$n,m \in \Natpos$, can be solved in $O(2^n \mul (n + m))$ time.
\qed
\end{proof}
If a problem can be solved in $O(2^n \mul (n + m))$ time, then it can be 
solved in $O(2^{2 \mul n} \mul m)$ time.
This justifies the statement that the problem mentioned in 
Theorem~\ref{theorem-complexity-shortest-const-1} can be solved in time 
exponential in $n$ and linear in $m$ on a RAM.

Thirdly, we consider the problem of determining whether an arbitrary 
instruction sequence from $\ISbr$ whose length is $\len(\TSTNZp{n})$ 
plus an amount that depends linearly on $n$ correctly implements the 
function $\tstnzc{n}$, for $n > 0$. 
The following theorem states that this problem is co-NP-complete.
\begin{theorem}
\label{theorem-complexity-shortest-linear}
Let $q \in \Rat$ be such that $q > 0$ 
and $m \in \Nat$ be such that $m > 3$. 
Then the problem of deciding whether an $X \in \ISbr$ such that 
$\len(X) \leq \len(\TSTNZp{n}) + \ceil{q \mul n} + m$ computes 
$\tstnzc{n}$, for $n \in \Natpos$, is co-NP-complete.
\end{theorem}
\begin{proof}
We call this problem $\CORRECTnzt$.
Let $b_1,\ldots,b_n \in \Bool$ be such that execution of $X$ with 
$b_1 \ldots b_n$ as the initial contents of $\inbr{1},\ldots,\inbr{n}$, 
respectively does not yield $\tstnzc{n}(b_1,\ldots,b_n)$ as final 
content of $\outbr$.
Then we can verify that $X$ does not compute $\tstnzc{n}$ by simulating 
the execution of $X$ with $b_1 \ldots b_n$ as the initial contents of 
$\inbr{1},\ldots,\inbr{n}$, respectively.
The simulation of this execution of $X$ takes $O(n)$ steps.
Hence, $\CORRECTnzt$ is in co-NP.
We still have to show that $\CORRECTnzt$ is co-NP-hard.

We show that $\CORRECTnzt$ is co-NP-hard by proving that $\coSAT$, i.e., 
the complement of $\SAT$, is reducible to $\CORRECTnzt$.
The proof is based on the idea that, for each proposition $P$ containing 
$n$ variables, an instruction sequence can be constructed of which the 
first part computes $\tstnzc{n'}$, where $n' \geq n$ and $n'$ depends 
linearly on the length of $P$, the second part computes the truth 
function $\funct{f}{\Bool^n}{\Bool}$ expressed by $P$, and the third 
part negates the result of the first part if the result of the second 
part is $0$. 
The length of a proposition $P$ is the number of occurrences of 
variables and connectives in $P$.
In this proof, we write $\len(P)$ for the length of~$P$.

For the construction outlined above, use is made of a polynomial-time 
computable function from the set of all propositions to $\ISbr$ that is 
introduced in the proof of Proposition~4 from~\cite{BM13a}.  
The function concerned associates a given proposition with an 
instruction sequence that computes the truth function expressed by the 
given proposition and whose length depends linearly on the length of the 
given proposition.
In~\cite{BM13a}, this function is denoted by 
$\mathit{inseq}_\mathrm{bf}$.
In this proof, we denote this function shortly by $\Phi$.

Let $c \in \Natpos$  be such that, for all propositions $P$,
$\len(\Phi(P)) < c \mul \len(P)$.
Let $\Phi^*$ be such that, for all propositions $P$, 
$\Phi(P) = \Phi^*(P) \conc \outbr.\setbr{1} \conc \halt$.
Let $\TSTNZs{n}$ be such that $\TSTNZp{n} = \TSTNZs{n} \conc \halt$, and 
let $\TSTNZss{n}$ be $\TSTNZs{n}$ with each occurrence of $\outbr$ in
$\TSTNZs{n}$ replaced by $\auxbr{1}$.
Let $c' \in \Natpos$ be such that $(c' - 1) \mul q < 1 < c' \mul q$.
Let $\Psi$ be the transformation from the set of all 
propositions to $\ISbr$ such that, for all $n \in \Natpos$, for all 
propositions $P$ with $\len(P) = n$,
$\Psi(P) =
 \TSTNZss{c \mul c' \mul n} \conc \Phi^*(P) \conc \fjmp{4} \conc 
 \ptst{\auxbr{1}.\getbr} \conc \outbr.\setbr{1} \conc \halt \conc
 \ptst{\auxbr{1}.\getbr} \conc \ptst{\outbr.\setbr{0}} \conc 
 \outbr.\setbr{1} \conc \halt$.
Because $\Phi$ is polynomial-time computable, $\Psi$ is  polynomial-time 
computable.

Let $n \in \Natpos$, and let $P$ be a proposition with $\len(P) = n$.
Then the part $\TSTNZss{c \mul c' \mul n}$ of $\Psi(P)$ reads 
$c \mul c' \mul n$ input registers. 
The part $\Phi^*(P)$ of $\Psi(P)$ does not read additional input 
registers because at most $n$ variables can occur in a proposition of 
length $n$.
Hence, $\Psi(P)$ reads $c \mul c' \mul n$ input registers, which is 
linear in~$n$.

Because $1 < c' \mul q$, we have that, for all $n \in \Natpos$,  
$c \mul n \leq q \mul c \mul c' \mul n$.
Using this, we find that, for all $n \in \Natpos$, for all propositions 
$P$ with $\len(P) = n$, 
$\len(\Psi(P)) =
 \len(\TSTNZp{c \mul c' \mul n}) + \len(\Phi(P)) + 5 \leq 
 \len(\TSTNZp{c \mul c' \mul n}) + c \mul n + m \leq
 \len(\TSTNZp{c \mul c' \mul n}) + \ceil{q \mul c \mul c' \mul n} + m$.
Hence, for all $n \in \Natpos$, for all propositions $P$ with 
$\len(P) = n$, $\Psi(P)$ is an instance of $\CORRECTnzt$ that reads 
$c \mul c' \mul n$ input registers.

Let $P$ be a proposition.
Then, because $\Phi$ associates $P$ with an instruction sequence that 
computes the truth function expressed by $P$, we have that
(a)~if $\coSAT(P)$, then execution of $\Psi(P)$ yields the correct final 
content of $\outbr$ for each initial content of the input registers and
(b)~if not $\coSAT(P)$, then execution of $\Psi(P)$ yields the 
complement of the the correct final content of $\outbr$ for some initial 
content of the input registers.
Hence, for all proposition $P$, $\coSAT(P)$ iff 
$\CORRECTnzt(\Psi(P))$.
\qed
\end{proof}
To make the problem mentioned in 
Theorem~\ref{theorem-complexity-shortest-linear} more concrete, we 
mention that the instruction sequence $\TSTNZ{n}$ from 
Section~\ref{sect-ISNZ-natural} is an instance of this problem.

In the proof of Theorem~\ref{theorem-complexity-shortest-linear}, we
could have used, instead of $\Psi$, a less obvious transformation in 
which no auxiliary registers are used at the cost of $\Phi^*(P)$ 
occurring twice in the transformation of proposition $P$.
In that case, the transformation would yield good instruction sequences
in the sense defined in the second next paragraph.

Below we present a theorem stating that the problem of determining 
whether an arbitrary instruction sequence from $\ISbr$ whose length is 
$\len(\TSTNZp{n})$ plus a constant amount $m$ correctly implements the 
function $\tstnzc{n}$, for $n > 0$, can be solved in 
$O(2^{6 \mul m} \mul (n + m) + n \mul (n + m))$ time on a RAM if the 
instruction sequences are restricted to good instruction sequences.

A \emph{good instruction sequence} is defined as an instruction sequence 
that is of the form $X \conc \outbr.\setbr{1} \conc \halt$, where
$X \in \ISbr$ is such that only read instructions and forward jump
instructions $\fjmp{l}$, where $l > 0$, occur in $X$.
A \emph{very good instruction sequence} is a good instruction sequence 
in which no input register name appears in more than one occurrence of a 
read instruction.
A \emph{good instruction sequence for} $\tstnzc{n}$ is a good 
instruction sequence $X$ such that 
$\iregs(X) = \set{i \in \Natpos \where i \leq n}$
and $X$ computes $\tstnzc{n}$.
A \emph{very good instruction sequence for} $\tstnzc{n}$ is a very good 
instruction sequence $X$ such that 
$\iregs(X) = \set{i \in \Natpos \where i \leq n}$
and $X$ computes $\tstnzc{n}$.

Notice that instructions concerning auxiliary registers do not occur in 
good instruction sequences and very good instruction sequences.

\begin{lemma}
\label{lemma-good-inseqs}
Let $n,k \in \Natpos$, and 
let $X$ be a very good instruction sequence for $\tstnzc{n}$ in which 
there are $k$ occurrences of $\fjmp{1}$ or $\fjmp{2}$.
Then $X$ can be transformed into a very good instruction sequence for 
$\tstnzc{n}$ that is at least $k-1$ primitive instructions shorter than 
$X$ and in which $\fjmp{1}$ does not occur and $\fjmp{2}$ occurs at most 
once.
\end{lemma}
\begin{proof}
Obviously, it is sufficient to prove that $X$ can be transformed into a 
very good instruction sequence for $\tstnzc{n}$ that is not longer than 
$X$ and in which $\fjmp{1}$ does not occur and $\fjmp{2}$ occurs at most 
once.

The following properties of $X$ are useful in the proof:
\begin{pplist}
\item[(a)]
$X$ is an instruction sequence in which, for each $j \leq n$ for which
the occurrence of $\ptst{\inbr{j}.\getbr}$ or $\ntst{\inbr{j}.\getbr}$ 
in $X$ is not the last occurrence of a read instruction in $X$, after 
execution of the occurrence of $\ptst{\inbr{j}.\getbr}$ or 
$\ntst{\inbr{j}.\getbr}$ in $X$, execution proceeds such that execution 
of the next occurrence of a read instruction in $X$ will take place if 
the initial content of $\inbr{j}$ is $0$;
\item[(b)]
$X$ is an instruction sequence in which, for each $j \leq n$, either 
$\ptst{\inbr{j}.\getbr}$ or $\ntst{\inbr{j}.\getbr}$ occurs;
\item[(c)]
$X$ is an instruction sequence in which three or more read instructions 
in a row do not occur;
\item[(d)]
$X$ can be transformed into a very good instruction sequence for 
$\tstnzc{n}$ that is not longer than $X$ and in which two or more 
forward jump instructions in a row do not occur;
\item[(e)]
an isolated read instruction of $X$ is of the form 
$\ptst{\inbr{i}.\getbr}$ 
if $n > 1$ and two or more forward jump instructions in a row do not 
occur;
\item[(f)]
a read instruction pair of $X$ is of the form 
$\ntst{\inbr{i}.\getbr} \conc \ptst{\inbr{j}.\getbr}$
if $n > 2$ and two or more forward jump instructions in a row do not 
occur;
\item[(g)]
\begin{plist}
\item[-]
if the number of isolated read instructions in $X$ is odd, then $X$ can 
be transformed into a very good instruction sequence for $\tstnzc{n}$ 
that is not longer than $X$ and has the form
$Y \conc \ptst{\inbr{i}.\getbr} \conc \outbr.\setbr{1} \conc \halt$, 
where no isolated read instructions occur in $Y$, if $n > 1$ and the 
form $\ptst{\inbr{1}.\getbr} \conc \outbr.\setbr{1} \conc \halt$ 
if $n = 1$;
\item[-]
if the number of isolated read instructions in $X$ is even, then $X$ can 
be transformed into a very good instruction sequence for $\tstnzc{n}$ 
that is not longer than $X$ and has the form
$Y \conc \ntst{\inbr{i}.\getbr} \conc \ptst{\inbr{j}.\getbr} \conc
 \outbr.\setbr{1} \conc \halt$, 
where no isolated read instructions occur in $Y$, if $n > 2$ and the
form 
$\ntst{\inbr{i}.\getbr} \conc \ptst{\inbr{j}.\getbr} \conc
 \outbr.\setbr{1} \conc \halt$ if $n = 2$. 
\end{plist}
\end{pplist}

Property~(a) holds because otherwise, on execution of $X$, the final 
content of $\outbr$ is $0$ in the case where the input register to be 
read by the next occurrence of a read instruction contains $1$.

Property~(b) holds because otherwise, on execution of $X$, the final 
content of $\outbr$ is $0$ in the case where, for some $j \leq n$ for 
which neither $\ptst{\inbr{j}.\getbr}$ nor $\ntst{\inbr{j}.\getbr}$ 
occurs in $X$, $\inbr{j}$ contains $1$.

That property~(c) holds can be seen by assuming that $X$ is of the form 
\begin{center}
$u_1 \conc u_2 \conc u_3 \conc X''$ or 
$X' \conc u_1 \conc u_2 \conc u_3 \conc X''$,
\end{center}
where $u_1$, $u_2$, and $u_3$ are read instructions by which 
respectively $\inbr{i}$, $\inbr{i'}$, and $\inbr{i''}$, for some 
$i, i', i'' \leq n$, are read.
Let $R$ be $\iregs(X')$ if $X$ is of the second form and $\emptyset$ 
otherwise.
On execution of $X$, if the initial content of all input registers 
$\inbr{j}$ with $j \in R$ is $0$ and the initial content of $\inbr{i}$ 
is $1$, then, by property~(a), $u_1 \conc u_2 \conc u_3 \conc X''$ is 
eventually executed and so execution of $u_3 \conc X''$ must yield $1$ 
as final content of $\outbr$.
However, on execution of $X$, if the initial content of all input 
registers is $0$, then, by property~(a), $u_3 \conc X''$ is eventually 
executed and must yield upon execution $0$ as final content of $\outbr$.
Hence, the assumption leads to a contradiction.

That property~(d) holds can be seen by assuming that $X$ is of one of 
the following forms:
\begin{center}
(i)~$\fjmp{l} \conc \fjmp{l'} \conc X''$, 
(ii)~$u \conc \fjmp{l} \conc \fjmp{l'} \conc X''$,
(iii)~$X' \conc u \conc \fjmp{l} \conc \fjmp{l'} \conc X''$, 
\end{center}
where $u$ is a read instruction by which $\inbr{i}$, for some 
$i \leq n$, is read.
Let $R$ be $\iregs(X') \union \set{i}$ if $X$ is of form~(iii),
$\set{i}$ if $X$ is of form~(ii), and $\emptyset$ otherwise, and
let $\funct{\alpha}{R}{\Bool}$ be such that $\alpha(j) = 0$ for all 
$j \in R$.
In case~(i), $\fjmp{l}$ and the $l-1$ next instructions can be removed 
from $X$.
In cases~(ii) and~(iii), in the case that $u$ is 
$\ptst{\inbr{i}.\getbr}$, on execution of $X$, if the initial content of 
$\inbr{i}$ and all input registers $\inbr{j}$ with $j \in R$ is $0$, 
then, by property~(a), $\fjmp{l'} \conc X''$ is eventually executed and 
so, by property~(b), $\fjmp{l'} \conc X''_\alpha$ computes $\tstnzc{n'}$ 
for some $n' < n$.
However, the suffix of $X''_\alpha$ obtained by removing its first 
$l'-1$ instructions computes the same function.
Hence, $\fjmp{l'}$ and the $l'-1$ next instructions can be removed from 
$X$ after replacing chains of forward jumps by direct jumps where 
needed because of removed jump instructions.
In cases~(ii) and~(iii), in the case that $u$ is 
$\ntst{\inbr{i}.\getbr}$, it can be similarly shown that $\fjmp{l}$ and 
the $l-1$ next instructions can be removed from $X$ after replacing 
chains of forward jumps by direct jumps where needed.

That property~(e) holds can be seen by assuming that $X$ is of the form 
\begin{center}
$\ntst{\inbr{i}.\getbr} \conc \fjmp{l'} \conc X''$ or  
$X' \conc \fjmp{l} \conc \ntst{\inbr{i}.\getbr} \conc \fjmp{l'}
 \conc X''$,
\end{center}
where $i \leq n$.
Let $R$ be $\iregs(X')$ if $X$ is of the second form and $\emptyset$ 
otherwise.
On execution of $X$, if the initial content of $\inbr{i}$ and all input 
registers $\inbr{j}$ with $j \in R$ is $0$, then, by property~(a), 
$\ntst{\inbr{i}.\getbr} \conc \fjmp{l'} \conc X''$ is eventually 
executed, but, after execution of its first instruction,
execution proceeds such that execution of the next occurrence of a read 
instruction in $X$ will not take place if $n > 1$ and $l' > 1$. 
Hence, the assumption leads to a contradiction with property~(a) if 
$l' > 1$. 
If instead $l' = 1$, then execution of the next occurrence of a read 
instruction in $X$ will take place, but the final content of $\outbr$
will not depend on the content of $\inbr{i}$.
Hence, the assumption leads to a contradiction with the fact that $X$
computes $\tstnzc{n}$ if $l' = 1$.

That property~(f) holds can be seen by assuming that $X$ is of one of 
the following forms:
\begin{center}
\begin{quote}
(i)~$\ptst{\inbr{i}.\getbr} \conc \ptst{\inbr{j}.\getbr} \conc
 \fjmp{l'} \conc X''$,
(ii)~$X' \conc \fjmp{l} \conc \ptst{\inbr{i}.\getbr} \conc
 \ptst{\inbr{j}.\getbr} \conc \fjmp{l'} \conc X''$,
(iii)~$\ptst{\inbr{i}.\getbr} \conc \ntst{\inbr{j}.\getbr} \conc
 \fjmp{l'} \conc X''$, 
(iv)~$X' \conc \fjmp{l} \conc \ptst{\inbr{i}.\getbr} \conc
 \ntst{\inbr{j}.\getbr} \conc \fjmp{l'} \conc X''$, 
(v)~$\ntst{\inbr{i}.\getbr} \conc \ntst{\inbr{j}.\getbr} \conc
 \fjmp{l'} \conc X''$, 
(vi)~$X' \conc \fjmp{l} \conc \ntst{\inbr{i}.\getbr} \conc
 \ntst{\inbr{j}.\getbr} \conc \fjmp{l'} \conc X''$, 
\end{quote}
\end{center}
where $i,j \leq n$.
Let $R$ be $\iregs(X')$ if $X$ is of form (ii), (iv) or~(vi) and 
$\emptyset$ otherwise, and
let $\funct{\alpha}{R}{\Bool}$ be such that $\alpha(j) = 0$ for all 
$j \in R$.
In cases~(i)--(iv), on execution of $X$, if the initial content of 
$\inbr{i}$ and all input registers $\inbr{j}$ with $j \in R$ is $0$, 
then, by property~(a),  
$\ptst{\inbr{i}.\getbr} \conc \ptst{\inbr{j}.\getbr} \conc \fjmp{l'}
 \conc X''$ (cases~(i)--(ii)) or
$\ptst{\inbr{i}.\getbr} \conc \ntst{\inbr{j}.\getbr} \conc \fjmp{l'}
 \conc X''$ (cases~(iii)--(iv)) 
is eventually executed but, after execution of the first instruction of
either instruction sequence, execution does not proceed with the 
execution of the next instruction.
Hence, in cases~(i)--(iv), the assumption leads to a contradiction 
with property~(a).
In cases~(v) and~(vi), on execution of $X$, if the initial content of 
$\inbr{i}$ and all input registers $\inbr{j}$ with $j \in R$ is $0$,  
then, by property~(a), 
$\ntst{\inbr{j}.\getbr} \conc \fjmp{l'} \conc X''$ is eventually 
executed and so, by property~(b), 
$\ntst{\inbr{j}.\getbr} \conc \fjmp{l'} \conc X''_\alpha$ is a very good 
instruction sequence for $\tstnzc{n'}$ for some $n'< n$.
Hence, in cases~(v) and~(vi), the assumption leads to a contradiction 
with property~(e) if $n > 2$ and two or more forward jump instructions 
in a row do not occur.

That property~(g) holds follows directly from properties (c)--(f), the 
fact that $X$ with each chain of two or more forward jump instructions 
in $X$ replaced by a single jump instruction that leads to the same 
instruction is also a very good instruction sequence for $\tstnzc{n}$, 
and the following claim:
if there are no chains of two or more forward jump instructions in $X$,
then
\begin{plist}
\item
if $X$ is of the form
$\ptst{\inbr{i_1}.\getbr} \conc 
 \fjmp{l} \conc \ntst{\inbr{i_2}.\getbr} \conc 
 \ptst{\inbr{i_3}.\getbr} \conc u \conc X''$, 
where $u$ is either a jump instruction or $\outbr.\setbr{1}$, then
$\ntst{\inbr{i_1}.\getbr} \conc 
 \ptst{\inbr{i_2}.\getbr} \conc \fjmp{l{-}1} \conc 
 \ptst{\inbr{i_3}.\getbr} \conc u \conc X''$
is also a very good instruction sequence for $\tstnzc{n}$;
\item
if $X$ is of the form
$X' \conc 
 \fjmp{l_1} \conc \ptst{\inbr{i_1}.\getbr} \conc 
 \fjmp{l_2} \conc \ntst{\inbr{i_2}.\getbr} \conc 
 \ptst{\inbr{i_3}.\getbr} \conc u \conc X''$, 
where $u$ is either a jump instruction or $\outbr.\setbr{1}$, then
$X' \conc 
 \fjmp{l_1} \conc \ntst{\inbr{i_1}.\getbr} \conc 
 \ptst{\inbr{i_2}.\getbr} \conc \fjmp{l_2{-}1} \conc 
 \ptst{\inbr{i_3}.\getbr} \conc u \conc X''$
is also a very good instruction sequence for $\tstnzc{n}$;
\item
if $X$ is of the form
$\ptst{\inbr{i_1}.\getbr} \conc 
 \fjmp{l} \conc \ptst{\inbr{i_2}.\getbr} \conc u \conc X''$, 
where $u$ is either a jump instruction or $\outbr.\setbr{1}$, then
$\ntst{\inbr{i_1}.\getbr} \conc
 \ptst{\inbr{i_2}.\getbr} \conc u \conc X''$
is also a very good instruction sequence for $\tstnzc{n}$;
\item
if $X$ is of the form
$X' \conc 
 \fjmp{l_1} \conc \ptst{\inbr{i_1}.\getbr} \conc 
 \fjmp{l_2} \conc \ptst{\inbr{i_2}.\getbr} \conc u \conc X''$, 
where $u$ is either a jump instruction or $\outbr.\setbr{1}$, then
$X' \conc 
 \fjmp{l_1} \conc \ntst{\inbr{i_1}.\getbr} \conc 
 \ptst{\inbr{i_2}.\getbr} \conc u \conc X''$
is also a very good instruction sequence for $\tstnzc{n}$.
\end{plist}
This claim is easy to check by case distinction on the content of the 
input registers involved, using the fact that, in order to compute 
$\tstnzc{n}$, all occurrences of jump instructions in $X$ immediately
following an occurrence of a read instruction of the form 
$\ptst{\inbr{i}.\getbr}$ must lead to the last but one instruction of
$X$.

Notice the following concerning properties~(a) and~(b): 
for each of the properties~(c)--(g), property~(a) is used to show that 
it holds and, for each of the properties~(d) and~(f), property~(b) is 
used in addition to show that it holds.  

It follows directly from properties~(c)--(g) that:
\begin{plist}
\item
if $n$ is even, then $X$ can be transformed into a very good instruction 
sequence for $\tstnzc{n}$ that is not longer than $X$ and has the form 
$Y \conc \outbr.\setbr{1} \conc \halt$, 
where $Y$ consists of read instruction pairs separated by instructions 
of the form $\fjmp{l}$ with $l \geq 2$;
\item
if $n$ is odd and $n > 1$, then $X$ can be transformed into a very good 
in\-struction sequence for $\tstnzc{n}$ that is not longer than $X$ and 
has the form 
$Y \conc \fjmp{l'} \conc \ptst{\inbr{i}.\getbr} \conc
 \outbr.\setbr{1} \conc \halt$, 
where $l' \geq 2$ and $Y$ consists of read instruction pairs separated 
by instructions of the form $\fjmp{l}$ with $l \geq 2$;
\item
if $n = 1$, then $X$ can be transformed into a very good instruction 
sequence for $\tstnzc{n}$ that is not longer than $X$ and has the form 
$\ptst{\inbr{1}.\getbr} \conc \outbr.\setbr{1} \conc \halt$.
\end{plist}
However, for the instructions of the form $\fjmp{l}$ separating read 
instruction pairs, we have that $l \neq 2$, because the final content 
of $\outbr$ may be independent of the initial content of some input 
registers if $l = 2$.
Moreover, for the instruction of the form $\fjmp{l'}$ immediately 
preceding an isolated read instruction that is the last occurrence of
a read instruction, we have that $l' = 2$, because the final content 
of $\outbr$ may be independent of the initial content of some input 
register if $l' = 1$ and simply wrong if $l' \geq 3$.
Hence, $X$ can be transformed into a very good instruction sequence for 
$\tstnzc{n}$ that is not longer than $X$ and in which $\fjmp{1}$ does 
not occur and $\fjmp{2}$ occurs at most once.
\qed
\end{proof}

\begin{lemma}
\label{lemma-multiple-reads}
Let $n,m \in \Natpos$, and 
let $X$ be a good instruction sequence for $\tstnzc{n}$ with 
$\len(X) = \len(\TSTNZp{n}) + m$. 
Then there are less than $6 \mul m$ input register names that appear in 
more than one occurrence of a read instruction in $X$.
\end{lemma}
\begin{proof}
In the case that $6 \mul m > n$, the result is immediate.
The case that $6 \mul m \leq n$ is proved by contradiction.
Suppose that there are $6 \mul m$ input register names that appear in 
more than one occurrence of a read instruction in $X$. 
Let $R$ be the set of the numbers of these $6 \mul m$ input registers.
Let $\alpha_0$ be the function from $R$ to $\Bool$ defined by 
$\alpha_0(i) = 0$ for all $i \in R$.
The following properties of $X_{\alpha_0}$ are obvious:
(a)~$X_{\alpha_0}$ is a very good instruction sequence for 
$\tstnzc{n - 6 \mul m}$,
(b)~$X_{\alpha_0}$ contains at least $12 \mul m$ occurrences of 
$\fjmp{1}$ or $\fjmp{2}$.
By Lemma~\ref{lemma-good-inseqs}, it follows from these properties that 
$X_{\alpha_0}$ can be transformed into a very good instruction sequence 
for $\tstnzc{n - 6 \mul m}$, say $Y$, that is at least $12 \mul m - 1$ 
primitive instructions shorter than $X_{\alpha_0}$.
So, $\len(Y) \leq \len(\TSTNZp{n}) + m - (12 \mul m - 1)$.
Consider the case that $n$ is even and the case that $n$ is odd.
In both cases, it is easy to calculate that 
$\len(\TSTNZp{n}) + m - (12 \mul m - 1) < \len(\TSTNZp{n - 6 \mul m})$.
In the case that $n$ is even, using that $1 < 2 \mul m$ for 
$m \in \Natpos$, the calculation goes as follows: 
$\len(\TSTNZp{n}) + m - (12 \mul m - 1) = 3 \mul n / 2 - 11 \mul m + 2 <
 3 \mul n / 2 - 9 \mul m + 1 = \len(\TSTNZp{n - 6 \mul m})$.
The calculation for the other case goes similarly.
From $\len(Y) \leq \len(\TSTNZp{n}) + m - (12 \mul m - 1)$ and
$\len(\TSTNZp{n}) + m - (12 \mul m - 1) < \len(\TSTNZp{n - 6 \mul m})$
it follows that $\len(Y) < \len(\TSTNZp{n - 6 \mul m})$.
However, by Theorem~\ref{theorem-NZTISp-shortest}, 
$\len(Y) \geq \len(\TSTNZp{n - 6 \mul m})$. 
Hence, a contradiction.
\qed
\end{proof}

\begin{lemma}
\label{lemma-division-correct}
Let $n \in \Natpos$,
let $X \in \ISbr$ be such that 
$\iregs(X) = \set{i \in \Natpos \where i \leq n}$, and
let $R \subset \iregs(X)$.  
Then $X$ computes $\tstnzc{n}$ if: 
(i)~for the unique function $\alpha_0$ from $R$ to $\Bool$ such that 
$\alpha_0(i) = 0$ for all $i \in R$, $X_{\alpha_0}$ computes 
$\tstnzc{n - \card(R)}$;
(ii)~for each function $\alpha$ from $R$ to $\Bool$ such that 
$\alpha(i) = 1$ for at least one $i \in R$, $X_\alpha$ yields upon 
execution $1$ as final content of $\outbr$ for each combination of 
initial contents of the input registers whose numbers belong to 
$\iregs(X_\alpha)$.
\end{lemma}
\begin{proof}
This follows directly from the definition of $\tstnzc{n}$.
\qed
\end{proof}

\begin{lemma}
\label{lemma-very-good-always-1}
The problem of deciding whether an $X \in \ISbr$ such that $X$ is a very
good instruction sequence and $\len(X) = n$ yields upon execution $1$ as 
final content of $\outbr$ for each combination of initial contents of 
the input registers whose numbers belong to $\iregs(X)$, for 
$n \in \Natpos$, can be solved in $O(n)$ time on a RAM.
\end{lemma}
\begin{proof}
Let $X^k$, for $1 \leq k \leq n$, be the suffix of $X$ whose length is 
$k$.
Clearly, for $k > 2$, $X^k$ is a very good instruction sequence.
Moreover, we know that execution of $X^1$ never yields $1$ as final 
content of $\outbr$ and execution of $X^2$ always yields $1$ as final 
content of $\outbr$.
Below, we use the convention that $X^k = \halt$ for all $k \leq 0$.

Let $k > 1$ and $l > 0$.
Then: 
\begin{plist}
\item
if $X^{k+1} = \ptst{\inbr{i}.\getbr} \conc X^k$, then execution of 
$X^{k+1}$ always yields $1$ as final content of $\outbr$ iff both 
execution of $X^k$ and execution of $X^{k-1}$ always do so;
\item
if $X^{k+1} = \ntst{\inbr{i}.\getbr} \conc X^k$, then execution of 
$X^{k+1}$ always yields $1$ as final content of $\outbr$ iff both 
execution of $X^k$ and execution of $X^{k-1}$ always do so;
\item
if $X^{k+1} = \fjmp{l} \conc X^k$, then execution of $X^{k+1}$ always 
yields $1$ as final content of $\outbr$ iff execution of $X^{k+1-l}$ 
always does so.
\end{plist}
The first two implications from left to right are easily proved by a 
case distinction on the content of $\inbr{i}$.
It is trivial to prove the last implication from left to right and the 
three implications from right to left. 

Now, let $S$ be the set of all $k \geq 1$ for which $X^k$ always yields 
$1$ as final content of $\outbr$.
Then, using the above bi-implications, the $k$'s which belong to $S$ can 
be determined in increasing order as follows:
\begin{plist}
\item
$1 \notin S$;
\item
$2 \in S$;
\item
if $X^{k+1} = \ptst{\inbr{i}.\getbr} \conc X^k$, then $k+1 \in S$ iff 
$k \in S$ and $k-1 \in S$;
\item
if $X^{k+1} = \ntst{\inbr{i}.\getbr} \conc X^k$, then $k+1 \in S$ iff 
$k \in S$ and $k-1 \in S$;
\item
if $X^{k+1} = \fjmp{l} \conc X^k$ and $l \leq k$, then $k+1 \in S$ iff 
$k+1-l \in S$;
\item
if $X^{k+1} = \fjmp{l} \conc X^k$ and $l > k$, then $k+1 \notin S$.
\end{plist}
Execution of $X$ always yields $1$ as final content of $\outbr$ iff  
$n$ is the last $k$ for which it can be determined in this way that it 
belongs to $S$. 
Clearly, this can be determined in $O(n)$ time if the set $S$ is 
represented by its characteristic function.
\qed
\end{proof}

\begin{lemma}
\label{lemma-very-good-correct}
The problem of deciding whether an $X \in \ISbr$ such that $X$ is a very
good instruction sequence, $\len(X) = \len(\TSTNZp{n}) + m$, and
$\iregs(X) = \set{i \in \Natpos \where i \leq n}$ computes $\tstnzc{n}$, 
for $n,m \in \Natpos$, can be solved in $O(n \mul (n + m))$ time on a 
RAM.
\end{lemma}
\begin{proof}
Consider the following procedure:
\par
\begin{pplist}
\item[step 1:] 
determine whether execution of $X$ yields $0$ as the final content of 
$\outbr$ if the initial content of each input register whose numbers 
belong to $\iregs(X)$ is $0$; if this is the case, then go on with 
step~2; otherwise $X$ does not compute $\tstnzc{n}$ and we are finished;
\item[step 2:] 
determine, for each $i \in \iregs(X)$, whether execution of 
$X[\inbr{i} = 1]$ yields $1$ as the final content of $\outbr$ for each 
combination of initial contents of the input registers whose numbers 
belong to $\iregs(X) \diff \set{i}$; if 
this is the case, then $X$ computes $\tstnzc{n}$ and we are finished;
otherwise $X$ does not compute $\tstnzc{n}$ and we are finished.
\end{pplist}
Clearly, together steps~1 and~2 cover all combinations of the initial 
contents of the input registers whose numbers belong to $\iregs(X)$.
Step~1 can be done in $O(n + m)$ time.
By Lemma~\ref{lemma-very-good-always-1}, per $i$, step~2 can also be 
done in $O(n + m)$ time.
Because $\card(\iregs(X)) = n$, step~2 as a whole can be done in 
$O(n \mul (n + m))$ time.
Consequently, the whole procedure can be done in $O(n \mul (n + m))$ 
time.
\qed
\end{proof}

\begin{theorem}
\label{theorem-complexity-shortest-const-2}
The problem of deciding whether an $X \in \ISbr$ such that $X$ is a good 
instruction sequence, $\len(X) = \len(\TSTNZp{n}) + m$, and
$\iregs(X) = \set{i \in \Natpos \where i \leq n}$ computes $\tstnzc{n}$, 
for $n,m \in \Natpos$, can be solved in 
$O(2^{6 \mul m} \mul (n + m) + n \mul (n + m))$ time on a RAM.
\end{theorem}
\begin{proof}
Consider the following procedure:
\par
\begin{pplist}
\item[step 1:] 
determine the subset $R$ of $\set{k \in \Natpos \where k \leq n}$ such 
that $i \in R$ iff $\inbr{i}$ appears in more than one occurrence of a 
read instruction in $X$; if $\card(R) < 6 \mul m$, then go on with 
step~2; otherwise, by Lemma~\ref{lemma-multiple-reads}, $X$ does not 
compute $\tstnzc{n}$ and we are finished; 
\item[step 2:] 
for the unique function $\alpha_0$ from $R$ to $\Bool$ such that 
$\alpha_0(i) = 0$ for all $i \in R$, determine whether $X_{\alpha_0}$ 
computes $\tstnzc{n - \card(R)}$; if this is the case, then go on with 
step~3; otherwise $X$ does not compute $\tstnzc{n}$ and we are finished;
\item[step 3:] 
for each function $\alpha$ from $R$ to $\Bool$ such that $\alpha(i) = 1$ 
for at least one $i \in R$, determine whether $X_\alpha$ yields upon 
execution $1$ as final content of $\outbr$ for each combination of 
initial contents of the input registers whose numbers belong to 
$\iregs(X_\alpha)$; if this is the case, then, by 
Lemma~\ref{lemma-division-correct}, $X$ computes $\tstnzc{n}$ and we are 
finished; otherwise $X$ does not compute $\tstnzc{n}$ and we are 
finished.
\end{pplist}
Clearly, step~1 can be done in $O(n + m)$ time.
By their construction, $X_{\alpha_0}$ and all $X_\alpha$'s are very good 
instruction sequences.
Hence, by Lemma~\ref{lemma-very-good-correct}, step 2 can be done in 
$O(n \mul (n + m))$ time and, by Lemma~\ref{lemma-very-good-always-1}, 
per $\alpha$, step~3 can be done in $O(n + m)$ time.
Because there may be $2^{6 \mul m} - 1$ $\alpha$'s, step~3 as a whole 
can be done in $O(2^{6 \mul m} \mul (n + m))$ time.
Consequently, the whole procedure can be done in 
$O(2^{6 \mul m} \mul (n + m) + n \mul (n + m))$ time.
\qed
\end{proof}
If a problem can be solved in 
$O(2^{6 \mul m} \mul (n + m) + n \mul (n + m))$ time, then it can be 
solved in $O(2^{7 \mul m} \mul n^2)$ time.
This justifies the statement that the problem mentioned in 
Theorem~\ref{theorem-complexity-shortest-const-2} can be solved in time 
quadratic in $n$ and exponential in $m$ on a RAM.
It is an open question whether
Theorem~\ref{theorem-complexity-shortest-const-2} goes through if the 
restriction to good instruction sequences is dropped.

The following result is a corollary of the proof of 
Theorem~\ref{theorem-complexity-shortest-linear} and the remark about 
that proof made following it.
\begin{corollary}
\label{corollary-complexity-shortest-linear}
Let $q \in \Rat$ be such that $q > 0$ and $m \in \Nat$. 
Then the problem of deciding whether an $X \in \ISbr$ such that 
$X$ is a good instruction sequence and
$\len(X) \leq \len(\TSTNZp{n}) + \ceil{q \mul n} + m$ computes 
$\tstnzc{n}$, for $n \in \Natpos$, is co-NP-complete.
\end{corollary}
So, the problem of determining whether an instruction sequence from 
$\ISbr$ whose length is $\len(\TSTNZp{n})$ plus an amount that depends 
linearly on $n$ correctly implements the non-zeroness test function 
$\tstnzc{n}$ remains co-NP-complete if the instruction sequences are 
restricted to good instruction sequences.
It can easily be shown that this problem can be solved in time 
polynomial in $n$ on a RAM if the instruction sequences are restricted 
to very good instruction sequences.
Hence, if no input register name may appear in more than one occurrence 
of a read instruction, then we get a better result. 

\section{Concluding Remarks}
\label{sect-concl}

Within the context of finite instruction sequences that contain only 
instructions to set and get the content of Boolean registers, forward 
jump instructions, and a termination instruction, we have investigated 
under what restrictions on these instruction sequences the correctness 
of an arbitrary instruction sequence as an implementation of the 
restriction to $\Bool^n$ of the function from $\Bool^*$ to $\Bool$ that 
models the non-zeroness test function on natural numbers with respect to 
their binary representations, for $n > 0$, can be efficiently 
determined. 
We expect that results similar to the main results established for this 
function, i.e., Theorems~\ref{theorem-complexity-shortest},
\ref{theorem-complexity-shortest-linear}, 
and~\ref{theorem-complexity-shortest-const-2}, can be established for 
many other functions, but also that finding such results is a 
challenging problem.
To our knowledge, the idea of looking for such results is new.

An important step in establishing the main results has been the 
determination of a shortest instruction sequence that correctly 
implements the function that models the non-zeroness test on natural 
numbers less than $2^n$ (cf.\ Theorem~\ref{theorem-NZTISp-shortest}).
In~\cite{BM13a}, an approach to computational complexity is presented in
which instruction sequence size is used as complexity measure.
The step just mentioned, provides a lower bound (in fact the greatest 
lower bound) for the instruction sequence size complexity of this 
function.
Moreover, it provides answers to concrete questions like ``what is the 
length of the shortest instruction sequence that correctly implements 
the function that models the non-zeroness test on natural numbers less 
than $2^{64}$?''.

The work presented in this paper concerns the question to what extent it 
can be efficiently determined, for $n > 0$, whether the restriction to 
$\Bool^n$ of a given function from $\Bool^*$ to $\Bool$ is correctly 
implemented by an arbitrary instruction sequence from a set by which, 
for all $m > 0$, all functions from $\Bool^m$ to $\Bool$ can be 
computed.
To our knowledge there is no previous work related to programming that 
addresses a question similar to this one.
Because each function from $\Bool^m$ to $\Bool$ can be computed by a 
Boolean circuit as well, there is of course the question to what extent 
it can be efficiently determined, for $n > 0$, whether the restriction 
to $\Bool^n$ of a given function from $\Bool^*$ to $\Bool$ is correctly 
implemented by an arbitrary Boolean circuit.

There are computational problems concerning Boolean circuits that can be 
viewed as correctness problems.
For example, the complement of the satisfiability problem for Boolean 
circuits (see e.g.~\cite{Sip13a}) can be viewed as the problem of 
determining whether an arbitrary Boolean circuit correctly implements 
the restriction to $\Bool^n$ of the function from $\Bool^*$ to $\Bool$ 
whose value is constantly $0$, for $n > 0$.
However, as far as we know, there is no work on Boolean circuits that 
pays attention to limitations on Boolean circuits under which problems 
concerning Boolean circuits that can be viewed as correctness problems 
can be solved in polynomial time. 
In other words, to our knowledge, there is no work on Boolean circuits 
that addresses a question related to the one addressed in this paper.

In~\cite{BM14e}, it was shown that, for the parity function, shortest
correct instruction sequences require the use of auxiliary registers.
This is not the case for the function that models the non-zeroness test
(cf.\ Corollary~\ref{corollary-correct-aux}).
In~\cite{BM14e}, in addition to the commands $\setbr{0}$, $\setbr{1}$, 
and $\getbr$, the command $\negbr$ is used.
This command serves for complementing the content of an auxiliary 
register.
In~\cite{BM15a}, where $\negbr$ is denoted by $\mathrm{c/c}$, it is 
shown that this command is not needed for functional completeness, but 
that its addition gives sometimes rise to shorter instruction sequences.
In the current paper, using $\negbr$ would have only one consequence: 
it would make it possible to drop the restriction to $m > 3$ in 
Theorem~\ref{theorem-complexity-shortest-linear}.

\section*{Acknowledgements}

We thank two anonymous referees for carefully reading a preliminary 
version of this paper, for pointing out several flaws in it, and for 
suggesting improvements of the presentation.

\bibliographystyle{splncs03}
\bibliography{IS}

\begin{thebibliography}{10}
\providecommand{\url}[1]{\texttt{#1}}
\providecommand{\urlprefix}{URL }

\bibitem{AHU74a}
Aho, A.V., Hopcroft, J.E., Ullman, J.D.: The Design and Analysis of Computer
  Algorithms. Addison-Wesley, Reading, MA (1974)

\bibitem{BL02a}
Bergstra, J.A., Loots, M.E.: Program algebra for sequential code. Journal of
  Logic and Algebraic Programming  51(2),  125--156 (2002)

\bibitem{BM13a}
Bergstra, J.A., Middelburg, C.A.: Instruction sequence based non-uniform
  complexity classes. Scientific Annals of Computer Science  24(1),  47--89
  (2014)

\bibitem{BM14a}
Bergstra, J.A., Middelburg, C.A.: On algorithmic equivalence of instruction
  sequences for computing bit string functions. Fundamenta Informaticae
  138(4),  411--434 (2015)

\bibitem{BM14e}
Bergstra, J.A., Middelburg, C.A.: Instruction sequence size complexity of
  parity. Fundamenta Informaticae  149(3),  297--309 (2016)

\bibitem{BM15a}
Bergstra, J.A., Middelburg, C.A.: On instruction sets for {Boolean} registers
  in program algebra. Scientific Annals of Computer Science  26(1),  1--26
  (2016)

\bibitem{BM17a}
Bergstra, J.A., Middelburg, C.A.: Axioms for behavioural congruence of
  single-pass instruction sequences. Scientific Annals of Computer Science
  27(2),  111--135 (2017)

\bibitem{CR73a}
Cook, S.A., Reckhow, R.A.: Time bounded random access machine. Journal of
  Computer and System Sciences  7(4),  354--375 (1973)

\bibitem{SiteIS}
Middelburg, C.A.: Instruction sequences as a theme in computer science. {\tt
  https:\linebreak[2]//instructionsequence.wordpress.com/} (2015)

\bibitem{Sip13a}
Sipser, M.: Introduction to the Theory of Computation. Cengage Learning,
  Boston, MA, third edn. (2013)

\bibitem{Weg05a}
Wegener, I.: Complexity Theory -- Exploring the Limits of Efficient Algorithms.
  Springer-Verlag, Berlin (2005)

\end{thebibliography}

\end{document}